\documentclass[11pt, onecolumn]{article}

\usepackage[letterpaper, dvips, nohead, verbose, margin=2.5cm]{geometry}
\usepackage[utf8]{inputenc}

\usepackage{amsmath}
\usepackage{amsthm,,mathtools}
\usepackage{amssymb}
\usepackage{thmtools}
\usepackage{thm-restate}
\usepackage{enumitem}
\usepackage{dsfont}
\usepackage{xspace}
\usepackage{comment}
\usepackage[ruled,noend,resetcount,linesnumbered]{algorithm2e}
\usepackage[colorlinks]{hyperref}
\usepackage[capitalise]{cleveref}

\setitemize{itemsep=0pt,topsep=2pt,parsep=0pt,partopsep=0pt}
\setenumerate{itemsep=0pt,topsep=2pt,parsep=0pt,partopsep=0pt}

\newcommand{\dd}{\mathinner{.\,.}}
\newcommand{\ceil}[1]{\lceil #1 \rceil}

\newcommand{\sub}{\subseteq}
\newcommand{\sm}{\setminus}
\newcommand{\dash}{\text{-}}

\newcommand{\Oh}{\mathcal{O}}
\newcommand{\Ohtilde}{\tilde{\Oh}}

\newcommand{\eps}{\varepsilon}

\newtheorem{theorem}{Theorem}[section]
\newtheorem{corollary}[theorem]{Corollary}
\newtheorem{proposition}[theorem]{Proposition}
\newtheorem{lemma}[theorem]{Lemma}
\newtheorem{fact}[theorem]{Fact}

\newtheorem{observation}[theorem]{Observation}
\newtheorem{problem}[theorem]{Problem}
\theoremstyle{definition}
\newtheorem{definition}[theorem]{Definition}
\crefname{problem}{Problem}{Problems}
\theoremstyle{remark}
\newtheorem{remark}[theorem]{Remark}

\newcommand{\R}{\mathbb{R}}

\newcommand{\Zz}{\mathbb{Z}_{\ge 0}}

\newcommand{\laned}{\mathsf{Lan\text{-}Ed}}
\newcommand{\commentout}[1]{}

\newcommand{\calL}{{\mathcal L}}
\newcommand{\calF}{{\mathcal F}}
\newcommand{\calP}{{\mathcal P}}

\newcommand{\calN}{{\mathcal N}}

\newcommand{\dist}{\mathrm{d}}

\usepackage{eqparbox,array}

\title{Language Edit Distance and Scored Parsing: Faster Algorithms \& Connection to Fundamental Graph Problems}
\author{Tomasz Kociumaka\thanks{Institute for Computer Science, Artificial Intelligence and Technology}, Barna Saha\thanks{University of California San Diego, This note replaces \cite{saha2015language}.}}
\date{}

\begin{document}

\maketitle

\newcommand{\G}{\mathcal{G}}
\renewcommand{\S}{\mathcal{S}}
\renewcommand{\P}{\mathcal{P}}
\newcommand{\Rz}{\mathbb{R}_{\ge 0}}
\newcommand{\Rc}{\overline{\mathbb{R}}_{\ge 0}}
\newcommand{\der}[1]{\stackrel{#1}{\leadsto}}

\begin{abstract}
    Given a context free language $\mathcal{L(G)}$ over alphabet $\Sigma$ and a string $s \in \Sigma^*$, {\em the language edit distance}  problem seeks the minimum number of edits (insertions, deletions and substitutions) required to convert $s$ into a valid member of $\mathcal{L(G)}$. The well-known dynamic programming algorithm solves this problem in $O(n^3)$ time (ignoring grammar size) where $n$ is the string length [Aho, Peterson 1972, Myers 1985]. Despite its numerous applications, to date there exists no algorithm that computes exact or approximate language edit distance problem in true subcubic time.

In this paper we give the first such algorithm that approximates language edit distance in subcubic time. For any arbitrary $\epsilon > 0$, our algorithm runs in $\tilde{O}(\frac{n^{2.491}}{\epsilon^2})$ time and returns an estimate within a multiplicative approximation factor of $(1+\epsilon)$. Moreover, an additive $\varepsilon n$ approximation can be computed in $\Oh(\frac{n^2}{\epsilon^{0.825}})$ time.  

To complement our upper bound result, we show that exact computation of language edit distance with insertion-only edits in truly subcubic time will imply a truly subcubic algorithm for all-pairs shortest paths which is a long-standing open question in computer science. 
\end{abstract}
\section{Introduction}
Given a model for data semantics and structures, estimating how well a dataset fits the model is a core question in large-scale data analysis. Formal languages (e.g., regular language, context free language) provide a generic framework for data modeling. The deviation from a model is measured by the least changes required in data to perfectly fit the model.  Aho and Peterson studied this basic question more than forty years back to design error-correcting parsers for programming languages. Given a context free grammar (CFG) $\calL(G)$  over alphabet $\Sigma$ and a string $s \in \Sigma^*$, they proposed {\em the language edit distance} problem, which determines the fewest number of edits (insertions, deletions and substitutions) along with an edit script to convert $s$ into a valid member of $\calL(G)$. Due to its fundamental nature, the language edit distance problem has found many applications in compiler optimization 
(error-correcting parser \cite{ap72,m95}), data mining and data management (anomaly detection, mining and repairing data quality problems \cite{kssy:13,gm:2013,p82}), computational biology (biological structure prediction \cite{ps:101,vgf:13}), machine learning (learning topic and behavioral models  \cite{Johnson:2010,Moore:2002,video3}), and signal processing (video and speech analysis \cite{speech}).

For the language edit distance problem, the proposed algorithm by Aho and Peterson has a running time of $O(|G|^2n^3)$ \cite{ap72}. The dependency on grammar size in run time was later improved by Myers to $O(|G|n^3)$ \cite{m95}. Naturally, the cubic time-complexity on string length is a major bottleneck for most applications. Except minor polylogarithmic improvements over $O(n^3)$ \cite{vgf:13,rn14, zakov2010}, when the paper was written, no truly subcubic algorithms for the language edit distance problem was known. Since then, significant progress has been made in this area (notably see \cite{BGSW2019}).

In this paper, we make several contributions.

\noindent {\bf Upper Bound.}
 We give the first truly subcubic algorithm to approximate language edit distance for arbitrary context free languages.  Our algorithm runs in $\Oh(|G|\frac{n^{2.491}}{\varepsilon^2})$ time and computes the language edit distance within a multiplicative approximation factor of $(1+\varepsilon)$ for any $\varepsilon > 0$. Moreover, we get an $\varepsilon n$-additive approximation with a running time of $\tilde{\Oh}(|G|\frac{n^2}{\varepsilon^{0.825}})$. 

 The above result is obtained by reducing language edit distance problem to computing $(min,+)$ matrix products. Moreover, we observe that the underlying matrices have multiple structural properties which help to speed up the computation.
 
 In particular, the above bounds hold for a more generic problem of {\em scored parsing} with bounded score. In a scored parsing problem, each production in a grammar is associated with a score. The score of a parsing tree is the sum of scores of the productions used at every node of the parsing tree. The goal here is to compute a parsing that minimizes the total score of the parse tree. Aho a Peterson showed that language edit distance computation can be reduced to the scored parsing problem where scores are either $0$ or $1$.

We note that by a lower bound result of Lee \cite{l97}, it is known that a faster context free grammar parsing (distinguishing between $0$ and nonzero edit distance) leads to a faster algorithm for boolean matrix multiplication. Therefore, obtaining any multiplicative approximation factor for language edit distance in $o(n^\omega)$ time is unlikely where $\omega$ is the exponent of fast matrix multiplication.


\noindent {\bf Lower Bound.} In a pursuit to explain the difficulty in obtaining exact subcubic algorithms, {\it we show that a subcubic algorithm for Stochastic context free grammar parsing parsing or computing language edit distance (with only insertion) will culminate in a breakthrough result for several fundamental graph problems.} In particular, we show  {\it any subcubic algorithm for the SCFG parsing  leads to a subcubic algorithm for the all-pairs shortest paths problem (APSP).}  Similarly, if language edit distance where only insertion is allowed as edit operation has a subcubic algorithm, so does APSP. This establishes surprising connection to these problems with a fundamental graph problem for which obtaining a subcubic algorithm is a long-standing open question.

Our lower bound results build upon a construction given by Lee \cite{l97} who showed a faster algorithm for CFG parsing implies a faster algorithm for boolean matrix multiplication.  For SCFG parsing, by suitably modifying his construction, we show a subcubic SCFG parser implies a subcubic $(\min, \times)$ matrix product computation \footnote{ $(\min, \times)$ matrix product $C$ of two matrices $A$ and $B$ of suitable dimensions is defined as $C[i,j]=\min_{k} A[i,k] \times B[k,j]$.}. Next, we prove a subcubic algorithm for $(\min, \times)$ matrix product leads to a subcubic algorithm for negative triangle detection in weighted graphs. Negative triangle detection is one of the many problems known to be subcubic equivalent with all-pairs shortest path problem \cite{williams2010}. Our second reduction interestingly also shows computing $(\min, +)$ product of matrices with real weights bounded by $\log{n}$ in subcubic time is unlikely to exist. In contrast, $(\min, +)$ product of matrices with integer weights bounded by $\log{n}$ can be done fast in $O(n^\omega \log{n})$ time. Our lower bound result for language edit distance also uses similar construction and builds upon it to additionally handle edit distance.

By subcubic equivalence \cite{williams2010,abboud:14}, a subcubic algorithm for language edit distance or SCFG parsing implies a subcubic algorithm for a large number of graph problems, e.g. detecting {\em minimum weight triangle}, {\em minimum weight cycle}, checking {\em metricity}, finding {\em second shortest path}, {\em replacement path}, {\em radius problem}.

\subsection{Related Works}
\label{section:related}

The language edit distance problem is a significant generalization of the widely-studied {\em string edit distance problem} where two strings need to be matched with minimum number of edits. The string edit distance problem can be exactly computed in quadratic time. Despite many efforts, a sub-quadratic exact algorithm for string edit distance does not exist. A recent result by Backurs and Indyk explains this difficulty by showing a sub-quadratic algorithm for string edit distance implies sub-exponential algorithm for satisfiability \cite{DBLP:journals/corr/BackursI14}. Our lower bound results are in a similar spirit which connects subcubic algorithm for language edit distance, and SCFG parsing to graph problems for which obtaining exact subcubic algorithms are long-standing open questions. For approximate string edit distance computation, there is a series of works that tried to lower the running time to near-linear  \cite{Sahinalp:1996,Cole:2002,lms:98,bjkk:04,bfc:06,ak:09,ako:10}. 

Language recognition and parsing problems have been studied for variety of languages under different models for decades \cite{mmn:stoc10,cckm:focs10,kls:mfcs11,akns:siam01,prr:rand03}. The works of \cite{mmn:stoc10,cckm:focs10,kls:mfcs11} study the
complexity of recognizing Dyck language in space-restricted streaming model. Alon, Krivelevich, Newman and Szegedy consider testing
regular language and Dyck language recognition problem using sub-linear queries \cite{akns:siam01}, followed by improved bounds in works of \cite{prr:rand03}.
The early works of $O(n^3)$ time algorithm for parsing context free grammars (such as the Cocke-Younger-Kasami algorithm (CYK or CKY algorihtm)) was improved by an elegant work of Valiant who obtained the first subcubic algorithm for context free grammar parsing \cite{Val1975}. For stochastic grammar parsing and language edit distance computation, till date there does not exist any true subcubic algorithm, except for minor polylogarithmic improvements in the running time \cite{vgf:13,rn14, zakov2010}.

\section{Preliminaries}
A \emph{(context-free) grammar} is a tuple $\G = (\S,\Sigma,\P,S)$, where
$\S$ is a set of \emph{symbols},
$\Sigma\sub \S$ is a set of \emph{terminals},
$\P\sub (\S\sm \Sigma)\times \S^*$ is a set of \emph{productions},
and $S\in\S\sm \Sigma$ is the \emph{starting symbol}.
Traditionally, $A\to R$ denotes a production $(A,R)\in \P$.
The \emph{size} of a grammar $\G$ is defined as $|\G| = |\S|+\sum_{(A,R)\in \P} |R|$.

For $X,Y\in \S^*$, we say that $Y$ can be \emph{derived} from $X$, denoted $X\leadsto Y$,
if $Y$ can be obtained from $X$ by repeatedly selecting a production $A\to R$ from $\P$ and replacing a (single occurrence of) $A$ with $R$.
The \emph{language} of a symbol $A\in \S$ is defined with $L(A)=\{X\in \Sigma^* : A \leadsto X\}$,
and the language of the grammar $\G$ is defined with $L(\G)=L(S)$.

\begin{problem}[Parsing]\label{prb:p}
    Given a grammar $\G=(\S,\Sigma,\P,S)$ of size $g$ and a string $X\in \Sigma^n$, decide if $X\in L(\G)$.
\end{problem}

\begin{theorem}[Valiant~\cite{Val1975}]\label{thm:valiant}
    The parsing problem can be solved in $\Oh(g^{\Oh(1)}n^{\omega})$ time.
\end{theorem}

In a \emph{scored grammar}, each production is associated with a non-negative real score;
formally, $\P \sub (\S \sm \Sigma) \times \S^* \times \Rz$.
Moreover, we write $X\der{s}Y$ if there is a finite sequence of intermediate strings $X=Z_1,\ldots,Z_{k}=Y$
and a sequence of productions $A_i\stackrel{s_i}{\to} R_i$, with $i\in [1\dd k)$,
such that $s=\sum_{i=1}^{k-1}s_i$ and $Z_{i+1}$ can be obtained from $Z_i$ by replacing a single occurrence of $A_i$ with $R_i$.
For a symbol $A\in \S$ and a string $X\in \Sigma^*$, we define the \emph{score of parsing $X$ into $A$}
as $s_A(X) = \inf\{s\in  \Rz : A \der{s} X\}$.
We assume that  $s_A : \Sigma^*\to \Rc$, where $\Rc = \Rz\cup\{\infty\}=\Rz\cup\{\inf \emptyset\}$.
Furthermore, we set $s_\G(X)=s_S(X)$.

The language $L(A)$ of a symbol $A\in \S$ can be defined by discarding the production scores  
and interpreting $\G$ as a standard context-free grammar or, equivalently, with
 $L(A) = \{X\in \Sigma^* : s_A(X) < \infty\}$. Similarly, $L(\G)=\{X\in \Sigma^* : s_\G(X) < \infty\}$.
Conversely, a standard grammar can be interpreted as a scored grammar by associating cost $0$ to each production.
Consequently, the scored parsing problem and the approximate scored parsing problem, defined below, both generalize \cref{prb:p}.

\begin{problem}[Scored Parsing]\label{prb:sp}
Given a scored grammar $\G=(\S,\Sigma,\P,S)$ of size $g$ and a string $X\in \Sigma^n$, compute the score $s_\G(X)$.
\end{problem}

\begin{fact}[Folklore; see~\cite{Coc1969,Kas1966,Sak1961,You1967}]\label{fct:cyk}
    The scored parsing problem can be solved in $\Oh(g^{\Oh(1)}n^{3})$ time.
\end{fact}

\begin{problem}[Approximate Scored Parsing]\label{prb:asp}
    Given a scored grammar $\G=(\S,\Sigma,\P,S)$ of size $g$, a string $X\in \Sigma^n$, and a parameter $\eps\in (0,1)$,
    compute $\tilde{s}_\G(X)\in \Rc$ such that $s_\G(X)\le \tilde{s}_\G(X)\le (1+\eps)s_\G(X)$.
\end{problem}

\begin{theorem}\label{thm:main}
The approximate scored parsing problem can be solved in $\Ohtilde(\eps^{-2}g^{\Oh(1)}n^{2.491})$ time.
\end{theorem}
Note: If $\omega = 2$, then the running time is $\Ohtilde(\eps^{-4/3}n^{7/3})$.
In general, we have some trade-off between the exponent at $\eps$ and the exponent at $n$;
in the theorem, the exponent at $n$ is optimized.

\subsection{CNF Grammars}
A  grammar $\G$ is a \emph{CNF} grammar (in Chomsky Normal Form) if each production $A\to R$ satisfies the following conditions:
\begin{enumerate}[label=(\alph*)]
    \item the starting symbol $S$ does not occur in $R$,
    \item if $A\ne S$, then $|R|=2$,
    \item if $A = S$, then $|R|\le 1$.
\end{enumerate}
Moreover, a grammar $\G$ is an \emph{almost-CNF} grammar if each production $A\to R$ satisfies $|R|\le 2$.

We say that two scored grammars $\G,\G'$ with the same terminals $\Sigma$ are \emph{equivalent} if $s_\G(X)=s_{\G'}(X)$ holds for all $X\in \Sigma^*$.
\begin{lemma}[\cite{BGSW2019}]\label{lem:almost}
    Given an almost-CNF scored grammar $\G$, an equivalent CNF scored grammar $\G'$ of size $|\G'|=|\G|^{\Oh(1)}$
    can be constructed in $\Oh(|\G|^{\Oh(1)})$ time.
\end{lemma}

\begin{corollary}\label{cor:cnf}
    Given a scored grammar $\G$,  an equivalent CNF grammar $\G'$ of size $|\G'|=|\G|^{\Oh(1)}$
    can be constructed in $\Oh(|\G|^{\Oh(1)})$ time.
\end{corollary}
\begin{proof}
    We discard every production $A\stackrel{s}{\to} B_1\cdots B_k$ with $k\ge 3$, 
    adding $k-2$ new symbols $C_2,\ldots,C_{k-1}$ and $k-1$ new productions  $C_2 \stackrel{0}{\to} B_1B_2$,
    $C_i \stackrel{0}{\to} C_{i-1}B_i$ for $i\in [3\dd k-1]$,
    and $A\stackrel{s}{\to} C_{k-1}B_k$. 
    It is easy to see that this process results in an equivalent almost-CNF scored grammar of size $\Oh(|\G|)$.
    We then apply \cref{lem:almost} to derive an equivalent CNF scored grammar.
\end{proof}

\section{CYK Parser}
In this section, we introduce some useful notation and, as a warm-up, derive a proof \cref{fct:cyk}.
Let us fix an instance of the scored parsing problem involving a CNF grammar $\G$.

For every symbol $A\in \S$, we define a matrix $M_A\in \Rc^{[0\dd n]\times [0\dd n]}$ 
so that $M_A[i,j]=s_A(X[i\dd j))$ if $0\le i\le j\le n$ and $M_A[i,j]=\infty$ if $0\le j < i \le n$.

The values $M_A[i,j]$ can be computed recursively using the following observation:

\begin{observation}\label{obs:cyk}
For a CNF scored grammar $\G$ and a string $X\in \Sigma^n$,
let $i,j\in [0\dd n]$ and $A\in \S$. 
If $A\in \Sigma$, then
\[ M_A[i,j] = \begin{cases}
0 & \text{if }j=i+1\text{ and }X[i]=A,\\
\infty & \text{otherwise.}
\end{cases}\]
If $A\in \S\sm (\Sigma\cup\{S\})$, then
\[ M_A[i,j] = \min_{A\stackrel{s}{\to}BC} \min_{k=i+1}^{j-1} (s + M_B[i,k]+M_C[k,j]).
\]
If $A=S$, then
\[M_A[i,j] = \begin{cases}
    \min_{S\stackrel{s}{\to}\eps} s & \text{if }i=j,\\
    \min_{S\stackrel{s}{\to}B} (s+M_B[i,j])&\text{otherwise.}
\end{cases}\]
\end{observation}
Consequently, if $\G$ is a CNF scored grammar, then the scored parsing problem can be solved in $\Oh(gn^3)$ time.
Combined with \cref{cor:cnf}, this immediately yields \cref{fct:cyk}.

\section{Valiant's Parser}
In this section, we present the parsing algorithm of Valiant~\cite{Val1975} in a modern interpretation by Okhotin~\cite{Okh2010}. While we state it for a scored grammar, we note that it does not yield any improvements upon \cref{fct:cyk}. However, introducing appropriate modifications, we derive a fast approximate scored parsing algorithm. 

\subsection{Notation}
We assume for simplicity that $n+1$ is a power of $2$; if this is not the case, we extend $X$ with arbitrary characters.
We can perform this transformation without loss of generality because the parser computes the whole score matrix $M_S$ so, in particular, we can retrieve the score of the original string.

For integers $\ell \in [0\dd \log (n+1)]$ and $p\in [0\dd \frac{n+1}{2^\ell})$,
we say that $[p\cdot 2^\ell \dd (p+1)\cdot 2^\ell)$ is a \emph{level-$\ell$ (dyadic) interval}.
For two intervals $I,J$, we write $I\prec J$ if $\max I < \min J$.
For a level-$\ell$ interval $I$ with $\ell>0$,
we denote by $I_L$ and $I_R$ the two level-$(\ell-1)$ intervals satisfying
$I = I_L\cup I_R$ and $I_L\prec I_R$.

For two sets $I,J$, let  $\Rc^{I\times J}$ be the family of matrices indexed by $I\times J$ with values in $\Rc$.
For a matrix $M\in \Rc^{[0\dd n]\times [0\dd n]}$ and sets $I,J\sub [0\dd n]$,
let  $M[I,J]\in \Rc^{I\times J}$ be the submatrix indexed by $I\times J$.

Given matrices $M\in\Rc^{I\times K}$ and $N\in \Rc^{K\times J}$,
we define the $(\min,+)$ product matrix $M\star N\in \Rc^{I\times J}$ so that 
$(M\star N)[i,j] = \min_{k\in K} (M[i,k]+N[k,j])$ for all $(i,j)\in I\times J$.
Moreover, for two matrices $M,N\in \Rc^{I\times J}$, we define $\min(M,N)$ to be the point-wise minimum of $M$ and $N$,
whereas for a matrix $M\in \Rc^{I\times J}$ and $s\in \Rc$, we define $s+M$ to be the matrix obtained from $M$ by adding $s$ to each entry.

\newcommand{\compute}{\mathsf{compute}}
\newcommand{\complete}{\mathsf{complete}}
\newcommand{\update}{\mathsf{update}}

\subsection{Implementation}
For each symbol $A\in \S$, the parser maintains a matrix $V_A\in \Rc^{[0\dd n]\times [0\dd n]}$,
aiming to eventually set $V_A := M_A$. 
This goal is achieved using three procedures specified below and implemented in \cref{alg:valiant}:
\begin{description}[leftmargin=2.7cm,style=sameline]
\item[$\compute(I)$:] Given a level-$\ell$ interval $I$, set $V_{A}[I,I]:=M_A[I,I]$ for all $A\in \S$.
\item[$\complete(I,J)$:] Given two level-$\ell$ intervals $I,J$ such that $I  \prec J$, set $V_{A}[I,J]:=M_A[I,J]$ for all $A\in \S$, assuming that the following holds for all $A\in \S$:
\begin{itemize}
   \item $V_A[I,I] = M_A[I,I]$,
   \item $V_A[J,J]=M_A[J,J]$,
   \item $V_A[I,J] = \min_{A\stackrel{s}{\to}BC} (s + M_B[I,K]\star M_C[K,J])$, where $K=(\max I \dd \min J)$.
\end{itemize} 
\item[$\update(I,K,J)$:] Given three level-$\ell$ intervals $I,K,J$ such that $I\prec K \prec J$
and $I\cup K$ or $K\cup J$ forms a level-$(\ell+1)$-interval,
set  $V_A[I,J] = \min(V_A[I,J], \min_{A\stackrel{s}{\to}BC} (s + V_B[I,K]\star V_C[K,J]))$
for all $A\in \S$.
\end{description}
In particular, note that executing $\compute([0\dd n])$ results in setting $V_A := M_A$ for all $A\in \S$.

\begin{algorithm}[H]
    \caption{Valiant's parser as interpreted by Okhotin.}\label{alg:valiant}
    \SetKwProg{Fn}{Procedure}{:}{}
    \Fn{$\compute(I)$}{
    \If{$|I|=1$}{
        Let $I = \{i\}$\;
        \lForEach{$A\in \S$}{$V_A[i,i]=\infty$}\label{ln:iinf}
        \lForEach{$S\stackrel{s}{\to} \eps$}{$V_S[i,i] = \min(V_S[i,i],s)$}\label{ln:iS}
    }\Else{
        $\compute(I_L)$\;\label{ln:L}
        $\compute(I_R)$\;\label{ln:R}
        \ForEach{$A\in \S$}{
            $V_A[I_R,I_L]=[\infty]^{I_R\times I_L}$\;\label{ln:iRL}
            $V_A[I_L,I_R]=[\infty]^{I_L\times I_R}$\;\label{ln:iLR}
        }
        $\complete(I_L,I_R)$\;\label{ln:complete}
    }
    }\BlankLine

    \Fn{$\complete(I,J)$}{
    \If{$|I|=|J|=1$}{
        Let $I = \{i\}$ and $J=\{j\}$\;
        \lIf{$j=i+1$}{$V_{X[i]}[i,j]=0$}\label{ln:zero}
        \lForEach{$S\stackrel{s}{\to} B$}{$V_S[i,j] = \min(V_S[i,j],s+V_B[i,j])$}\label{ln:S}
    }\Else{
        $\complete(I_R,J_L)$\;\label{ln:RL}
        $\update(I_L,I_R,J_L)$\;\label{ln:LRL}
        $\update(I_R,J_L,J_R)$\;\label{ln:RLR}
        $\complete(I_L,J_L)$\;\label{ln:LL}
        $\complete(I_R,J_R)$\;\label{ln:RR}
        $\update(I_L,J_L,J_R)$\;\label{ln:LLR}
        $\update(I_L,I_R,J_R)$\;\label{ln:LRR}
        $\complete(I_L,J_R)$\;\label{ln:LR}
    }
    }\BlankLine

    \Fn{$\update(I,K,J)$}{
    \lForEach{$A\stackrel{s}{\to}BC$}{
        $V_A[I,J] = \min(V_A[I,J], s + V_B[I,K]\star V_C[K,J])$
    }
    }\BlankLine

    $\compute([0\dd n])$\;
\end{algorithm}

\subsection{Correctness}

The implementation of $\update(I,J,K)$ is correct because it directly follows the specification.
The following two lemmas prove the correctness of $\complete(I,J)$ and $\compute(I)$.

\begin{lemma}\label{lem:complete}
\cref{alg:valiant} provides a correct implementation of $\complete(I,J)$.
\end{lemma}
\begin{proof}
First, suppose that $I=\{i\}$ and $J=\{j\}$ are level-$0$ intervals.
In this case, $K=(\max I\dd \min J) = [i+1\dd j-1]$.
By \cref{obs:cyk}, we have 
\[M_A[i,j]=\begin{cases}
    0 & \text{if }A=X[i]\text{ and }j=i+1,\\  
    \min_{S\stackrel{s}{\to} B} (s + M_B[i,j]) &\text{if }A=S,\\
    \min_{A\stackrel{s}{\to}BC} (s+M_B[I,K]\star M_C[K,J]) & \text{if }A\in \S \sm (\Sigma\cup \{S\}),\\
    \infty & \text{otherwise}.  
\end{cases}\]
In the last two cases, the precondition already guarantees $V_A[i,j] = M_A[i,j]$.
If $A=X[i]$ and $j=i+1$, then the algorithm correctly sets $V_A[i,j] := 0$ in \cref{ln:zero}.
Finally, at \cref{ln:S}, we already have $V_B[i,j]=M_B[i,j]$ for $B\in \S\sm \{S\}$,
so $V_S[i,j]=M_S[i,j]$ is also set correctly.

Next, suppose that $I$ and $J$ and level-$\ell$ intervals such that $\ell>0$ and $I\prec J$.
\begin{itemize}
\item Prior to the execution of \cref{ln:RL}, we have $V_A[I_R,I_R]=M_A[I_R,I_R]$,
 $V_A[J_L,J_L]=M_A[J_L,J_L]$, and
$V_A[I_R,J_L] = \min_{A\stackrel{s}{\to}BC} (s + M_B[I_R,K]\star M_C[K,J_L])$ for all $A\in \S$.
Hence, the preconditions for $\complete(I_R,J_L)$ are satisfied and, by induction, $V_A[I_R,J_L] := M_A[I_R,J_L]$ is set for all $A\in \S$.
\item Prior to the execution of \cref{ln:LRL}, we have $V_A[I_L,I_R]=M_A[I_L,I_R]$, $V_A[I_R,J_L]=M_A[I_R,J_L]$,
 and $V_A[I_L,J_L] = \min_{A\stackrel{s}{\to}BC} (s + M_B[I_L,K]\star M_C[K,J_L])$ for all $A\in \S$.
Hence, after the execution of \cref{ln:LRL}, $V_A[I_L,J_L] = \min_{A\stackrel{s}{\to}BC} (s + M_B[I_L,I_R\cup K]\star M_C[I_R\cup K,J_L])$ holds for all $A\in \S$.
\item Prior to the execution of \cref{ln:RLR}, we have $V_A[I_R,J_L]=M_A[I_R,J_L]$, $V_A[J_L,J_R]=M_A[J_L,J_R]$,
and $V_A[I_R,J_R] = \min_{A\stackrel{s}{\to}BC} (s + M_B[I_R,K]\star M_C[K,J_R])$ for all $A\in \S$.
Hence,  after the execution of \cref{ln:RLR}, $V_A[I_R,J_R] = \min_{A\stackrel{s}{\to}BC} (s + M_B[I_R,K\cup J_L]\star M_C[K\cup J_L,J_R])$ holds for all $A\in \S$.
\item Prior to the execution of \cref{ln:LL}, we have $V_A[I_L,I_L]=M_A[I_L,I_L]$, $V_A[J_L,J_L]=M_A[J_L,J_L]$, and
$V_A[I_L,J_L] = \min_{A\stackrel{s}{\to}BC} (s + M_B[I_L,I_R\cup K]\star M_C[I_R\cup K,J_L])$ for all $A\in \S$.
Hence, the preconditions for $\complete(I_L,J_L)$ are satisfied and, by induction, $V_A[I_L,J_L] := M_A[I_L,J_L]$ is set for all $A\in \S$.
\item Prior to the execution of \cref{ln:RR}, we have $V_A[I_R,I_R]=M_A[I_R,I_R]$, $V_A[J_R,J_R]=M_A[J_R,J_R]$,~and
$V_A[I_R,J_R] = \min_{A\stackrel{s}{\to}BC} (s + M_B[I_R,K\cup J_L]\star M_C[K\cup J_L,J_R])$ for all $A\in \S$.
Hence, the preconditions for $\complete(I_R,J_R)$ are satisfied and, by induction, $V_A[I_R,J_R] := M_A[I_R,J_R]$ set for all $A\in \S$.
\item Prior to the execution of \cref{ln:LLR}, we have $V_A[I_L,J_L]=M_A[I_L,J_L]$, $V_A[J_L,J_R]=M_A[J_L,J_R]$, and
$V_A[I_L,J_R] = \min_{A\stackrel{s}{\to}BC} (s + M_B[I_L,K]\star M_C[K,J_R])$ for all $A\in \S$.
Hence, after the execution of \cref{ln:LLR}, $V_A[I_L,J_R] = \min_{A\stackrel{s}{\to}BC} (s + M_B[I_L,K\cup J_L]\star M_C[K\cup J_L,J_R])$ holds for all $A\in \S$.
\item Prior to the execution of \cref{ln:LRR}, we have $V_A[I_L,I_R]=M_A[I_L,I_R]$, $V_A[I_R,J_R]=M_A[I_R,J_R]$,
and $V_A[I_L,J_R] = \min_{A\stackrel{s}{\to}BC} (s + M_B[I_L,K\cup J_L]\star M_C[K\cup J_L,J_R])$ for all $A\in \S$.
Hence,  after the execution of \cref{ln:LRR}, $V_A[I_L,J_R] = \min_{A\stackrel{s}{\to}BC} (s + M_B[I_L,I_R\cup K\cup J_L]\star M_C[I_R\cup K\cup J_L,J_R])$ holds for all $A\in \S$.
\item Prior to the execution of \cref{ln:LR}, we have $V_A[I_L,I_L]=M_A[I_L,I_L]$, $V_A[J_R,J_R]=M_A[J_R,J_R]$,
and $V_A[I_L,J_R] = \min_{A\stackrel{s}{\to}BC} (s + M_B[I_L,I_R\cup K\cup J_L]\star M_C[I_R\cup K\cup J_L,J_R])$.
Hence, the preconditions for $\complete(I_L,J_R)$ are satisfied and, by induction, $V_A[I_L,J_R] := M_A[I_L,J_R]$ is set for all $A\in \S$.
\end{itemize}
We conclude that $\complete(I,J)$ correctly sets $V_A[I,J]=M_A[I,J]$ for all $A\in \S$.
\end{proof}

\begin{lemma}\label{lem:compute}
\cref{alg:valiant} provides is a correct implementation of $\compute(I)$.
\end{lemma}
\begin{proof}
First, suppose that $I=\{i\}$ is a level-$0$ interval.
By \cref{obs:cyk}, we have 
\[M_A[i,i]=\begin{cases}
    \min_{S\stackrel{s}{\to} \eps} s & \text{if }A=S,\\  
    \infty &\text{otherwise.}
\end{cases}\]
Hence, the algorithm correctly sets $V_A[i,i]=M_A[i,i]$.

Next, suppose that $I$ is a level-$\ell$ interval for $\ell > 0$.
By induction, we have $V_A[I_L,I_L]=M_A[I_L,I_L]$
and $V_A[I_R,I_R]=M_A[I_R,I_R]$ for all $A\in \S$ after the execution of \cref{ln:L,ln:R}, respectively.
Moreover, $V_A[I_R,I_L]=[\infty]^{I_R\times I_L}=M_A[I_R,I_L]$ is correctly set in \cref{ln:iRL} for all $A\in \S$.
Finally, setting $V_A[I_L,I_R]=[\infty]^{I_L\times I_R}$ in \cref{ln:iLR},
results in $V_A[I_L,I_R] = \min_{A\stackrel{s}{\to}BC} (s + M_B[I_L,\emptyset]\star M_C[\emptyset,I_R])$ for all $A\in \S$.
Due to $(\max I_L\dd \min I_R) = \emptyset$, this means the the preconditions for $\complete(I_L,I_R)$ are satisfied and, by \cref{lem:complete}, $V_A[I_L,I_R] = M_A[I_L,I_R]$ holds after the execution of \cref{ln:complete} for all $A\in \S$.
We conclude that $\compute(I)$ correctly sets $V_A[I,I]:=M_A[I,I]$ for all $A\in \S$.
\end{proof}

\subsection{Running Time}
Let us first analyze the number of calls to $\compute(I)$, $\complete(I,J)$, and $\update(I,K,J)$ with intervals of a given level $\ell\in [0\dd \log(n+1)]$.
Clearly, there are $\frac{n+1}{2^\ell}$ calls to $\compute(I)$, where $I$ is a level-$\ell$ interval,
and each such call performs $\Oh(g 4^\ell)$ operations, for a total of $\Oh(gn^2)$ across all levels.
Furthermore, there are $\frac12\cdot \frac{n+1}{2^\ell}\cdot (\frac{n+1}{2^\ell}-1)$ calls to $\complete(I,J)$, where $I,J$ are level-$\ell$ intervals, and each such call performs $\Oh(g)$ operations, for a total of $\Oh(gn^2)$ across all levels.
Finally, there are $2\cdot \frac{n+1}{2^{\ell+1}}\cdot (\frac{n+1}{2^{\ell+1}}-1)$ calls to $\update(I,K,J)$, where $I,K,J$ are level-$\ell$ intervals.
Using a naive implementation of the $(\min,+)$-product~$\star$, each such call performs $\Oh(g 8^\ell)$ operations,
for a total of $\Oh(gn^3)$ across all levels. This yields an alternative proof of \cref{fct:cyk}.

However, if the score of each production is $0$, then the entries of each matrix $V_A$ belong to $\{0,\infty\}$,
and the $(\min,+)$-product is equivalent to a Boolean product (with $0$ and $\infty$ interpreted as $\textbf{true}$
and $\textbf{false}$, respectively), so the multiplication of two $m\times m$ matrices costs $\Oh(m^\omega)$ time.
In this case, each call to $\update(I,K,J)$ takes $\Oh(g2^{\omega \ell})$ time, for a total of $\Oh(gn^\omega)$ across all levels, thus proving \cref{thm:valiant}.

\section{Generic Approximate Parser}
Consider the following modified implementation of $\update(I,K,J)$.

\begin{algorithm}[H]
    \caption{Our parser with parameters $\alpha\in \R_{\ge 1}$, $d\in \Rc$, and $\delta\in \Rz$. The procedures $\compute$ and $\complete$ are as in~\cref{alg:valiant} except that they operate on matrices $U$ instead~of~$V$.}\label{alg:generic}
    \SetKwProg{Fn}{Procedure}{:}{}
    \SetKwComment{Comment}{$\triangleright$\ \rm}{}
    
    \Fn{$\update(I,K,J)$}{
        \ForEach{$A\stackrel{s}{\to}BC$}{
            \ForEach{$(i,k)\in I\times K$}{
                \lIf{$U_B[i,k] > d$}{$U'_B[i,k]:=\infty$}
                \ElseIf{$I\cup K$ is a dyadic interval}{
                    Set $U'_B[i,k]$ so that $M_B[i,k]\le U'_B[i,k] \le \alpha (U_B[i,k]+\delta|I|)$\;\label{ln:mmult}
                }\Else{
                    Set $U'_B[i,k]$ so that $M_B[i,k]\le U'_B[i,k] \le U_B[i,k]+\delta |I|$\;\label{ln:madd}
                }
            }
            \BlankLine
            \ForEach{$(k,j)\in K\times J$}{
                \lIf{$U_C[k,j] > d$}{$U'_C[k,j]:=\infty$}
                \ElseIf{$K\cup J$ is a dyadic interval}{
                    Set $U'_C[k,j]$ so that $M_C[k,j]\le U'_C[k,j] \le  \alpha (U_C[k,j]+\delta |I|)$\;\label{ln:nmult}
                }\Else{
                    Set $U'_C[k,j]$ so that $M_C[k,j] \le U'_C[k,j]\le U_C[k,j]+\delta |I|$\;\label{ln:nadd}
                }
            }
            \BlankLine
            $U_A[I,J] := \min(U_A[I,J], s + U'_B[I,K]\star U'_C[K,J])$\;
        }
    }
\end{algorithm}

\newcommand{\V}{\overline{V}}

Next, we prove the approximation guarantees on the output of our parser.
For this, we consider simultaneously running \cref{alg:valiant,alg:generic} and compare the matrices $U_A$ maintained by our parser with the corresponding matrices $V_A$ maintained by Valiant's parser.
Our argument relies on a carefully designed inductive claim involving two potentials $\phi$ and $\psi$.
\begin{definition}
For an integer $t>0$, let $p_t$ be the multiplicity of $2$ in the prime factorization of $t$.
For every two integers $i,j\in \Zz$, we define
\[
    \phi(i,j) = \begin{cases} \max\limits_{t=i+1}^j p_t & \text{if } i < j,\\
                                0 & \text{otherwise},   \end{cases}\qquad\text{and}\qquad
    \psi(i,j) = \begin{cases}\left(\sum\limits_{t=i+1}^j 2^{p_t}\right) - 2^{\phi(i,j)} &\text{if } i < j,\\
        0 & \text{otherwise}.\end{cases}\]
\end{definition}

\begin{lemma}\label{lem:ind}
During a simultaneous execution of \cref{alg:valiant,alg:generic}, the following properties hold at all times for every $A\in \S$ and every $i,j\in [0\dd n]$:
\begin{enumerate}[label=(\alph*)]
    \item $U_A[i,j]\ge V_A[i,j]$,\label{it:lb}
    \item $\min(d,U_A[i,j])\le \alpha^{\phi(i,j)}(2\delta \psi(i,j) + V_A[i,j])$.\label{it:ub}
\end{enumerate}
\end{lemma}
\begin{proof}
First, we observe that \cref{alg:valiant,alg:generic} are oblivious to the matrices $V_A$ and $U_A$, respectively, meaning that the two algorithms share the same sequence of instructions reading and writing the entries of these matrices,
and it suffices to analyze each such instruction independently.

It is easy to see that instructions in \cref{ln:iinf,ln:iS,ln:iLR,ln:iRL,ln:zero} of \cref{alg:valiant}
result in the same values being set to the corresponding entries $U_A[i,j]$ and $V_A[i,j]$.
As for \cref{ln:S}, the values set at $U_S[i,j]$ and $V_S[i,j]$ 
are $\min_{S\stackrel{s}{\to} B} (s + U_B[i,j])$ and $\min_{S\stackrel{s}{\to} B} (s + V_B[i,j])$,
respectively.
By the inductive assumption and due to $s\ge 0$, the following holds for every production $S\stackrel{s}{\to} B$:
\begin{enumerate}[label=(\alph*)]
    \item $U_B[i,j]+s\ge V_B[i,j]+s$,
    \item $\min(d,U_B[i,j]+s)\le \min(d+s,U_B[i,j]+s) = \min(d,U_B[i,j]) +s \le \alpha^{\phi(i,j)}(2\delta \psi(i,j) + V_B[i,j])+s \le \alpha^{\phi(i,j)}(2\delta \psi(i,j) + V_B[i,j]+s)$.
 \end{enumerate}
Consequently, the claimed conditions~\ref{it:lb} and~\ref{it:ub} are satisfied for $U_S[i,j]$ and $V_S[i,j]$.

It remains to analyze the effects of the two implementations of $\update(I,K,J)$ in \cref{alg:valiant,alg:generic}.
Suppose that $I,K,J$ are level-$\ell$ intervals and that $K\cup J$ is a dyadic interval;
the argument is symmetric when $I\cup K$ is a dyadic interval.
Let us fix a production $A\stackrel{s}{\to} BC$ and indices $(i,j)\in I\times J$.
\cref{alg:valiant} sets $V_A[i,j] := \min(V_A[i,j], \min_{k\in K}(s+V_B[i,k]+V_C[k,j]))$,
while \cref{alg:generic} sets $U_A[i,j] := \min(U_A[i,j], \min_{k\in K}(s+U'_B[i,k]+U'_C[k,j]))$,
where $U'_B[i,k]$ and $U'_C[k,j]$ are obtained from $U_B[i,k]$ and $U_C[k,j]$ by appropriate rounding.

As for the lower bound, for every $i\in I$, $j\in J$, and $k\in K$,
we have $U'_B[i,k]+U'_C[k,j]\ge U_B[i,k] + U_C[k,j] = V_B[i,k]+V_C[k,j]$. 
Therefore,
$\min(U_A[i,j], s + \min_{k\in K}(U'_B[i,k]+U'_C[k,j])) \ge \min(V_A[i,j], s + \min_{k\in K}(V_B[i,k]+V_C[k,j]))$,
i.e.,~\ref{it:lb} is satisfied for the newly assigned values $U_A[i,j]$ and $V_A[i,j]$.

As for the upper bound, observe that every $k\in K$ satisfies $\min(d,U'_B[i,k]) \le U_B[i,k] + \delta 2^\ell$
and $\min(d,U'_C[k,j])\le \alpha(U_C[k,j]+\delta 2^\ell)$.
Moreover, $\phi(k,j)=p_{\min J} = \ell$
and $\phi(i,j)=\phi(i,k)\ge p_{\min K} \ge \ell+1$ because $K \prec J$ and $K\cup J$ is a dyadic interval.
Consequently,
\begin{multline*}\min(d,U'_B[i,k])\le U_B[i,k] + \delta2^\ell \le \alpha^{\phi(i,k)}\left(2\delta \psi(i,k) + V_B[i,k]\right) + \delta 2^\ell \le \\
\le \alpha^{\phi(i,j)}\left(\delta (2\psi(i,k)+2^{\ell}) + V_B[i,k]\right), \end{multline*}
and
\begin{multline*}\min(d,U'_C[k,j])\le \alpha (U_C[k,j] + \delta 2^\ell) \le \alpha^{1+\phi(k,j)}\left(2\delta \psi(j,k) + V_C[k,j] + \delta 2^\ell\right) \le \\
    \le \alpha^{\phi(i,j)}\left(\delta (2\psi(k,j)+2^{\ell}) + V_C[k,j]\right).\end{multline*}
Therefore, due to $\psi(i,j) = \psi(i,k)+2^{\phi(i,k)}+\psi(k,j)+2^{\phi(k,j)}-2^{\phi(i,j)}=
\psi(i,k)+\psi(k,j)+2^\ell$, we have
\begin{multline*}
    \min(d,s+U'_B[i,k]+U'_C[k,j])\le s + \min(d, U'_B[i,k])+\min(d,U'_C[k,j])\le \\ \le s+ \alpha^{\phi(i,j)}\left(\delta (2\psi(i,k)+2^{\ell}) + V_B[i,k]\right)+\alpha^{\phi(i,j)}\left(\delta (2\psi(k,j)+2^{\ell}) + V_C[k,j]\right) \le \\ \le  \alpha^{\phi(i,j)}\left(2\delta \psi(i,j)+s+V_B[i,k]+V_C[k,j]\right).
\end{multline*}
We conclude that 
\begin{multline*}\min\left(d, \min\left(U_A[i,j], \min_{k\in K}(s+U'_B[i,k]+U'_C[k,j])\right)\right) \le \\ \le \alpha^{\phi(i,j)}\left(2\delta \psi(i,j)+ \min\left(V_A[i,j], \min_{k\in K}(s+V_B[i,k]+V_C[k,j])\right)\right),\end{multline*}
i.e., that~\ref{it:ub} is also satisfied for the newly assigned values $U_A[i,j]$ and $V_A[i,j]$.
\end{proof}

\newcommand{\Mat}{\mathcal{M}}
\section{Multiplicative Approximation}
To get a multiplicative approximation, we instantiate \cref{alg:generic}
with the following code to obtain~$U'_B[I,K]$;
the matrix $U'_C[K,J]$ is obtained in an analogous way.

\begin{algorithm}[H]
    \ForEach{$(i,k)\in I\times K$}{
        \lIf{$U_B[i,k]>d$}{$U'_B[i,k]:=\infty$}
        \Else{
            $U'_B[i,k] := \delta |I| \cdot \ceil{\frac{U_B[i,k]}{\delta |I|}}$\;
            \If{$I\cup K$ is a dyadic interval}{
                $U'_B[i,k] := \alpha^{\ceil{\log _{\alpha} U'_B[i,k]}}$\;
            }
        }
    }\caption{Rounding $U'_B$ for multiplicative approximation.}
\end{algorithm}
In other words, we round the entries of $U_B[I,K]$ up to the nearest integer multiple of $\delta |I|$.
Moreover, if $I\cup K$ is dyadic, then we further round the entries of $U'_B[I,K]$ 
up to the nearest integer power of $\alpha$.
As a result, $U'_B$ has at most $1+\ceil{\log_{\alpha}(\ceil{\frac{d}{\delta |I|}})}$ distinct entries
if $I\cup K$ is dyadic, and at most $1+\ceil{\frac{d}{\delta |I|}}$ distinct entries otherwise.

Since we use symmetric code to obtain $U'_C$, the final task is to compute the $(\min,+)$ product 
of two matrices with at most $1+\ceil{\log_{\alpha}(\ceil{\frac{d}{\delta |I|}})}$ and $1+\ceil{\frac{d}{\delta |I|}}$ entries, respectively.

In the following, we develop a specialized procedure based on the fast multiplication of rectangular Boolean matrices, assuming that the product of an $x\times y$ matrix
with a $y\times z$ matrix can be computed in $\Mat(x,y,z)$ time.

First, we recall the classic equivalence between the Boolean product  $(\min,+)$ of matrices over $\{0,\infty\}$.

\begin{fact}\label{fct:bool}
Consider matrices $M\in \{0,\infty\}^{I\times K}$ and $N\in \{0,\infty\}^{K\times J}$.
The product $M\star N$ can be computed in $\Oh(\Mat(|I|,|K|,|J|))$ time.
\end{fact}

Next, we provide two algorithms for the case when one matrix is over $\{0,\infty\}$ and the other is over~$\Rc$.

\begin{lemma}\label{lem:inf-r}
Consider matrices $M\in \{0,\infty\}^{I\times K}$ and $N\in \Rc^{K\times J}$.
The product $M\star N$ can be computed in $\Oh(\Mat(|I|,|K|,|J|r))$ time, where $r$ is the number of distinct entries in $N$.
\end{lemma}
\begin{proof}
    Let $E=\{N[k,j] : (k,j)\in K\times J\}$ be the set of distinct entries in $N$.
    Let us define a matrix $N'\in \{0,\infty\}^{K\times (J\times E)}$ so that,
    for each $k\in K$, $j\in J$, and $e\in E$, we have 
    \[N'[k,(j,e)] = \begin{cases}
        0 & \text{if }N[k,j] = e,\\
        \infty & \text{otherwise}.
    \end{cases}\]
    Observe that the following holds for each $(i,j)\in I\times J$:
    \begin{multline*}(M\star N)[i,j] = \min_{k\in K}(M[i,k]+N[k,j]) = \min_{k\in K}\{N[k,j] : M[i,k]=0\} =\\
    \min_{k\in K, e\in E} \{e : M[i,k]=N'[k,(j,e)]=0\} = \min_{e\in E}\{e : (M\star N')[i,(j,e)]=0\}.\end{multline*}
    Consequently, the sought matrix $M\star N$ can be derived from $M\star N'$.
    The running time is $\Oh(|K||J|r + \Mat(|I|,|K|,|J|r)+|I||J|r)=\Oh(\Mat(|I|,|K|,|J|r))$,
    with the three terms representing the cost of constructing $N'$, $M\star N'$, and $M\star N$, respectively,
    and the product $M\star N'$ computed using \cref{fct:bool}.
\end{proof}

\begin{lemma}\label{lem:inf-any}
    Consider matrices $M\in \{0,\infty\}^{I\times K}$ and $N\in \Rc^{K\times J}$.
    For every integer $\tau\ge \log |K|$, the product $M\star N$ can be computed in time $\Oh(\Mat(|I|,|K|,|J|\tau)+|I| |K| |J|/\tau)$.
    \end{lemma}
    \begin{proof}
        For each column $j\in J$ of $N$, define a sorting permutation $\sigma_j : K\to [0\dd |K|)$
        such that $\sigma_j(k)\le \sigma_j(k')$ implies $N[k,j] \le N[k',j]$ for $k,k'\in K$.
        Now, let us construct a matrix $N'\in \{0,\infty\}^{K\times (J\times [0\dd \tau))}$ such that,
        for each $k\in K$, $j\in J$, and $t\in [0\dd \tau)$, we have:
        \[N'[k,(j,t)] = \begin{cases}
            0 &\text{if }\sigma_j(k)< (t+1)\tfrac{m}{\tau},\\
            \infty &\text{otherwise}.
        \end{cases}\]
        Let us fix $(i,j)\in I\times J$ such that $M[i,k]=0$ for some $k\in K$. Observe that
        \[(M\star N)[i,j] = N\left[\arg\min_{k\in K}\{\sigma_j(k) : M[i,k]=0\},j\right].\]
        Moreover, for every $t\in [0\dd \tau)$, we have 
        \[(M\star N')[i,(j,t)]=0 \quad\text{if and only if}\quad \min_{k\in K}\{\sigma_j(k) : M[i,k]=0\}<(t+1)\tfrac{m}{\tau}.\]
        Consequently, denoting $t_{i,j} = \min\{t\in [0\dd \tau) : (M\star N')[i,(j,t)]=0\}$,
        we obtain
        \[(M\star N)[i,j] = \begin{cases} \infty & \text{if }t_{i,j}=\infty,\\\min_{k\in K}\{N[k,j] : M[i,k]=0\text{ and }t_{i,j}\tfrac{m}{\tau}\le \sigma_j(k) < (t_{i,j}+1)\tfrac{m}{\tau}\} & \text{otherwise}.\end{cases}\]
        This allows computing $(M\star N)[i,j]$ by iterating over $k\in \sigma_j^{-1}\big(\big[\lceil t_{i,j}\tfrac{m}{\tau} \rceil\dd \lceil (t_{i,j}+1)\tfrac{m}{\tau}\rceil\big)\big)$.
        
        Overall, the cost of computing each permutation $\sigma_j$ and its inverse $\sigma^{-1}_j$ is $\Oh(|K|\log |K|)$. Based on these permutations, the matrix $N'$ is constructed in $\Oh(|K||J|(\log |K|+\tau))=\Oh(|K||J|\tau)$ time.
        Computing $M\star N'$ using \cref{fct:bool} takes $\Oh(\Mat(|I|,|K|,|J|\tau))$ time,
        and deriving $M\star N$ costs $\Oh(|I||J|\tau + |I||K||J|/\tau)$ time.
        The overall running time is $\Oh(\Mat(|I|,|K|,|J|\tau)+|I| |K| |J|/\tau)$.
    \end{proof}

Using \cref{lem:inf-r,lem:inf-any}, we derive our final algorithm for computing the $(\min,+)$ product.
\begin{proposition}\label{prp:minplus}
    Consider matrices $M\in \Rc^{I\times K}$ and $N\in \Rc^{K\times J}$.
    The product $M\star N$ can be computed in time $\Oh(\min(\Mat(|I|r_M,|K|,|J|r_N),\min_{\tau \ge \log |K|}(\Mat(|I|r_M,|K|,|J|\tau) + |I||K||J|r_M/\tau)))$,
    where $r_M$ and $r_N$ denote the number of distinct entries in $M$ and $N$, respectively.
\end{proposition}
\begin{proof}
    Let $E = \{M[i,k] : (i,k)\in I\times K\}$ be the set of distinct entries of $M$.
    Let us define a matrix $M'\in \{0,\infty\}^{(I\times E)\times K}$ so that,
    for each $i\in I$, $e\in E$, and $k\in K$, we have
    \[M'[(i,e),k]=\begin{cases}
        0&\text{if }M[i,j]=e,\\
        \infty & \text{otherwise};
    \end{cases}.\]
    Observe that the following holds for each $(i,j)\in I\times J$:
    \begin{multline*}(M\star N)[i,j] = \min_{k\in K}(M[i,k]+N[k,j]) = \min_{e\in E, k\in K}\{e+N[k,j] : M[i,k]=e\} =\\
        \min_{e\in E}\left\{(e+\min_{k\in K}\{N[k,j] : M[(i,e),k] = 0\}\right\}=\min_{e\in E}\{e + (M\star N')[i,(j,e)]\}.\end{multline*}
    Consequently, the sought matrix $M\star N$ can be derived from $M' \star N$.
    The running time is $\Oh(|I||K|r_M + \min(\Mat(|I|r_M,|K|,|J|r_N),\min_{\tau \ge \log |K|}(\Mat(|I|r_M,|K|,|J|\tau) + |I||K||J|r_M/\tau))+|I||J|r_M)$,
    with the three terms representing the cost of constructing $M'$, $M'\star N$, and $M\star N$, respectively,
    and the product $M\star N'$ computed using \cref{lem:inf-r} or~\ref{lem:inf-any}.
    The middle term dominates the time complexity.
\end{proof}

With an appropriate choice of parameters, our algorithm yields the following result:
\begin{theorem}\label{thm:internal}
There exists an algorithm that, given a CNF scored grammar $\G$ of size $g$, a string $X\in \Sigma^n$,
and two real parameters $d\in \Rz$ and $\eps\in (0,1)$,
takes $\Oh(\eps^{-2}gn^{2.491})$ time to compute a value $\tilde{s}_\G(X)\in \Rc$
such that:
\begin{itemize}
    \item if $s_\G(X)<\frac{d}{2}$, then $s_\G(X) < \frac{(1+\eps)d}{2}$;
    \item if $\frac{d}{2} \le s_\G(X) < d$, then  $s_\G(X) \le \tilde{s}_\G(X) \le (1+\eps)s_\G(X)$;
    \item if $s_\G(X)\ge d$, then $\tilde{s}_\G(X) \ge d$.
\end{itemize}
\end{theorem}
\begin{proof}
We run \cref{alg:generic} with $\delta = \frac{\eps d }{6n\log n}$ and $\alpha = 1+\frac{\eps}{6\log n}$,
and return $\tilde{s}_\G(X)=\min(d,U_S[0,n])$.
This way, $\alpha^{\phi(0,n)}\le \alpha^{\log n}\le \exp(\frac{\eps}{6})\le 1+\frac{6}{5}\cdot \frac{\eps}{6}=1+\frac{\eps}{5}$,
and $2\delta\psi(i,j) \le 2\delta n \log n \le \frac{\eps d}{3}$.

Consequently, \cref{lem:ind}, combined with the correctness of \cref{alg:valiant},
implies $\tilde{s}_\G(X) \le (1+\frac{\eps}{5})(\frac{\eps d}{3}+s_\G(X))$.
Now, if $s_\G(X)<\frac{d}{2}$,
then $\tilde{s}_\G(X)\le (1+\frac\eps5)(\frac{\eps d}3 + s_\G(X))
\le \frac{d}{2}(1+\frac\eps5)(1+\frac{2\eps}{3})\le \frac{(1+\eps)d}{2}$ holds as claimed.
Moreover, if $\frac{d}{2} \le s_\G(X) < d$, then
$s_\G(X) = \min(d, s_\G(X)) \le \min(d,U_S[0,n]) =  \tilde{s}_\G(X)\le (1+\frac\eps5)(\frac{\eps d}3 + s_\G(X)) \le s_\G(X)(1+\frac\eps5)(1+\frac{2\eps}{3}) \le (1+\eps)s_\G(X)$ holds as claimed.
Finally, if $s_\G(X) \ge  d$, then $ \tilde{s}_\G(X) = \min(d,U_S[0,n]) \ge \min(d,s_\G(X))=  d$ hold as claimed.

In each call to $\update(I,J,K)$ for level-$\ell$ intervals (of length $m:=2^\ell$),
the matrices have $\Oh(\frac{d}{\delta m})=\Oh(\frac{d n \log n}{\eps d m})=\Ohtilde(\frac{n}{\eps m})$
and $\Oh(\log_{\alpha} \frac{d n \log n}{\eps d m})=\Ohtilde(\log_\alpha 2) = \Ohtilde(\eps^{-1})$ distinct entries, respectively.
Hence, by \cref{prp:minplus}, the cost is \[\Ohtilde\left(g\min\left(\Mat(\tfrac{m}{\eps},m,\tfrac{n}{\eps}),\min_{\tau \ge 1}\left(\Mat(\tfrac{m}{\eps},m,\tau m)+\tfrac{m^3}{\eps\tau}\right)\right)\right).\]
The overall running time of our parser is therefore
\begin{multline*}
\Ohtilde\left(g\max_{1\le m \le n}\left(\tfrac{n^2}{m^2} \min\left(\Mat(\tfrac{m}{\eps},m,\tfrac{n}{\eps}),\min_{\tau\ge 1}\left(\Mat(\tfrac{m}{\eps},m,\tau m)+\tfrac{m^3}{\eps\tau}\right)\right)\right)\right)=\\
\Ohtilde\left(g\max_{1\le m \le n} \left(\tfrac{n^2}{m^2} \min\left(\Mat(\tfrac{m}{\eps},m,\tfrac{n}{\eps}),\tfrac{m^4}{n}\right)\right)\right)
=\Ohtilde\left(g\max_{1\le m \le n} \min\left(\tfrac{n^2}{m^2}\Mat(\tfrac{m}{\eps},m,\tfrac{n}{\eps}),nm^2\right)\right)=\\
\Ohtilde\left(\eps^{-2}g\max_{1\le m \le n}\min\left(\tfrac{n^2}{m^2}\Mat(m,m,n),nm^2\right)\right)=\Oh(\eps^{-2}gn^{2.491}).\qedhere\end{multline*}
\end{proof}

We derive \cref{thm:main} using binary search to approximate $s_\G(X)$.
Since the most costly production in the optimum derivation of $X$ has score between $\frac{1}{n}s_\G(X)$ and $s_\G(X)$, the algorithm of \cref{thm:internal} is applied $\Oh(\log(ng))=\Ohtilde(1)$ times.

\section{Additive Approximation for Bounded-Difference Grammars}

\begin{definition}
    We say that a scored grammar $\G$ is \emph{$W$-bounded-difference}
    if the following holds for every $A\in \S$,  $a\in \Sigma$ and $X\in \Sigma^+$:
    \[|s_A(X)-s_A(aX)|\le W\qquad \text{and}\qquad |s_A(X)-s_A(Xa)|\le W.\]
\end{definition}

To get an additive approximation, we instantiate \cref{alg:generic}
with the following steps:

\begin{algorithm}[H]
    $\nu := \frac{4W}{\delta |I|}$\;
    \ForEach{$(i,k)\in I\times K$}{
        $U''_B[i,k] := \min\{U_B[i',k']+W|i-i'|+W|k-k'| : (i',k')\in I\times K\}$\;
    }
    \BlankLine
    \ForEach{$(i,k)\in I\times K$}{
        $U'_B[i,k] := \max\{\frac{2W}{\nu}\ceil{\frac{\nu}{2W} U''_B[i',k']} : (i',k')\in I\times K\text{ such that }\ceil{\nu i'}=\ceil{\nu i}\text{ and }\ceil{\nu k'}=\ceil{\nu k}\}$\;
    }\caption{Rounding $U'_B$ for additive approximation.}
\end{algorithm}

In this case, we fix $\alpha=1$ and $d=\infty$.
As for the lower bound, note that 
$U'_B[i,k] \ge U''_B[i,k] \ge U_B[i',k'] + W|i-i'|+W|k-k'| \ge V_B[i',k'] + W|i-i'|+W|k-k'|
=U_B[i',k'] + W|i-i'|+W|k-k'| \ge U_B[i,k'] +W|k-k'| \ge U_B[i,k]$
holds by the assumption that $\G$ is $W$-bounded difference.
Moreover, we note that $U''_B[I,K]$ is a $W$-bounded difference matrix.

As for the upper bound,
note that $U'_B[i,k] \le \frac{2W}{\nu} \ceil{\frac{\nu}{2W}U''_B[i',k']} \le U''_B[i',k']+\frac{2W}{\nu}
\le U''_B[i,k] + W|i-i'|+W|k-k'| + \frac{2W}{\nu} \le U_B[i,k]+\frac{4W}{\nu} = U_B[i,k]+\delta |I|$.

The matrix $U'_C$ is constructed analogously and satisfies analogous bounds.

As far as the implementation of $\update(I,K,J)$ for intervals of length $2^\ell=m$ is concerned,
we have two options. One is to use a naive $\Oh(m^3)$-time algorithm.
The other is to observe that $U'_B$ and $U'_C$ can be collapsed into 
matrices of size $\Oh(\nu m)$ whose neighboring entries differ by $0$ or $\pm \frac{W}{\nu}$.
Thus, we can use an algorithm of \cite{BGSW2019}:
a randomized one in $\Oh((\nu m)^{\bar{\omega}})$,
where $\bar{\omega}<2.861$ for deterministic algorithms
and $\bar{\omega}<2.825$ if one allows randomization.

Setting $\delta = \frac{\eps n W}{2 \psi(n)}=\tilde{\Omega}(\eps W)$, we get an $\eps n W$-additive approximation.
The running time of a single $\update(I,K,J)$ call is
\[\Ohtilde\left(g\min\left(m^3, \epsilon^{-\bar{\omega}} \right)\right).\]
The overall running time of our parser is therefore
\[
\Ohtilde\left(g\max_{1\le m \le n}\left(\tfrac{n^2}{m^2} \min\left(m^3, \epsilon^{-\bar{\omega}} \right)\right)\right)=\Ohtilde(gn^2 \eps^{-\bar{\omega}/3}).\]
This is $\Ohtilde(gn^2 \eps^{-0.954})$ for deterministic algorithms
and $\Ohtilde(gn^2 \eps^{-0.942})$ for randomized algorithms.

\begin{remark}
If $W=\Oh(1)$ and all the costs are integers,
then we can use an $\Oh(m^{\bar{\omega}})$ time algorithm instead of the naive one;
this improves the overall running time to $\Ohtilde(gn^2 \eps^{2-\bar{\omega}})$,
which is $\Ohtilde(gn^2 \eps^{-0.861})$ for deterministic algorithms and 
$\Ohtilde(gn^2 \eps^{-0.825})$ if one allows randomization.
\end{remark}

\section{Lower Bound}
\label{section:lowerbound}
\subsection{Subcubic Reduction of APSP to Language Edit Distance}
In this section, we prove the following theorem.

\begin{theorem}
\label{thm:lowerbound}
Given a context-free grammar $G$, and a string $s \in \Sigma^*$, $|s|=n$, if the language edit distance problem with only insertion as allowable edit can be solved in $O(|G|n^{3-\delta})$ time then that implies an $\tilde{O}(m^{3-\delta/3}\log{W})$ algorithm for all-pairs shortest path problem on weighted digraphs with $m$ vertices and maximum edge weight $W$.
\end{theorem}

We recall (redefine) some of the definitions first.

\noindent \textbf{Grammars \& Derivations.}
A context-free grammar (grammar for short) is a $4$-tuple $G=(\calN,\Sigma, \calP, S)$ where $\calN$ and $\Sigma$ are finite disjoint collection of nonterminals and terminals respectively. $\calP$ is the set of productions of the form $A \rightarrow \alpha$ where $A \in \calN$ and $\alpha \in (\calN \cup \Sigma)^*$. $S$ is a distinguished {\em start} symbol in $\calN$.

For two strings $\alpha, \beta \in (\calN\cup\Sigma)^{*}$, we say $\alpha$ directly derives $\beta$, written as $\alpha\Rightarrow \beta$, if one can write $\alpha=\alpha_1A\alpha_2$ and $\beta=\alpha_1\gamma\alpha_2$ such that $A \rightarrow \gamma \in \calP$. Thus, $\beta$ is a result of applying the production $A \rightarrow \gamma$ to $\alpha$.

$\calL(G)$ is the context-free language generated by grammar $G$, i.e., $\calL(G)=\{w \in \Sigma^* \mid S \xRightarrow{\ast}w\}$, where $\xRightarrow{\ast}$ implies that $w$ can be derived from $S$ using one or more production rules. If we always first expand the leftmost nonterminal during derivation, we have a leftmost derivation. Similarly, one can have a rightmost derivation. If $s \in \calL(G)$ it is always possible to have a leftmost (rightmost) derivation of $s$ from $S$.

We only consider grammars for which all the nonterminals are reachable, that is each of them is included in at least one derivation of a string in the language from $S$. Any unreachable nonterminal can be easily detected and removed decreasing the grammar size.

\noindent \textbf{Chomsky Normal Form (CNF).}
We consider the CNF representation of $G$. This implies every production is either of type (i) $A \rightarrow BC$,~ $A,B,C \in \calN$, or (ii) $A \rightarrow x,~ x \in \Sigma$ or (iii) $S \rightarrow \varepsilon$ if $\varepsilon \in \calL(G)$. It is well-known that every context-free grammar has a CNF representation. CNF representation is popularly used in many algorithms, including CYK and Earley's algorithm for CFG parsing \cite{Knuth:1997}. Prior works on cubic algorithms for language edit distance computation use CNF representation as well \cite{ap72,m95}

\begin{definition}[Language Edit Distance]
Given a grammar $G=(\calN,\Sigma, \calP, S)$ and $s \in \Sigma^*$, the language edit distance between $G$ and $s$ is defined as $d_{G}(G,s)=\min_{z: \in \calL(G)} \dist_{ed}(s,z)$ where $\dist_{ed}$ is the standard edit distance (insertion, deletion and substitution) between $s$ and $z$. If this minimum is attained by considering $z \in \calL(G)$, then $z$ serves as an witness for $\dist_{G}(G,s)$.
\end{definition}

We will often omit the subscript from $\dist_G$ and $\dist_{ed}$ and use $\dist$ to represent both language and string edit distance when that is clear from the context. We assume $\calL(G) \neq \phi$ and $\varepsilon \in \calL(G)$ so that $\dist_{G}(G,s)\leq |s|$.

\begin{definition}[$t$-approximation for Language Edit Distance]
Given a grammar $G=(\calN,\Sigma, \calP, S)$ and $s \in \Sigma^*$, a $t$-approximation algorithm for language edit distance problem, $t \geq 1$, returns a string $s'$ such that $s' \in \calL(G)$ and $\dist_{G}(G,s) \leq \dist_{ed}(s',s)\leq t* \dist_{G}(G,s)$.
\end{definition}

Note that to prove Theorem~\ref{thm:lowerbound} we only allow insertion as edits. Therefore, given a grammar $G$ and a string $s$, $|s|=n$, we compute $s' \in \calL(G)$ such that $s'$ can be obtained from $s$ by minimum number of insertions edits on $s$. If no such $s'$ exists in $\calL(G)$, then the language edit distance $d(s,G)$  is $\infty$.

We first define the output of a language edit distance algorithm rigorously. We use a notion of minimum consistent derivation. This is similar to the notion of consistent derivation used by Lee \cite{l97} to establish the lower bound for context free parsing. We need to additionally handle distance during parsing.

\begin{definition}
Given a context free grammar $G=(\mathcal{N},\Sigma, \calP, S)$, and a string $s \in \Sigma^*$. A nonterminal $A \in \calN$ mc-derives (minimally and consistently derives) $s_i^j$ if and only if the following condition holds:

1) If $A$ derives $s_{i}^{j}$ with a minimum score $l$, implying if $A(s)$ is the set of all strings that $A$ derives, $\min_{s'\in A(s)}\{dist_{ed}(s',s_{i}^{j})\}=l$, and

2. There is a derivation sequence $S \xRightarrow{\ast} s_{1}^{i-1}As_{j+1}^{n}$.
\end{definition}

\begin{definition}
A $\laned$~is an algorithm that takes a CFG $G=(\calN,\Sigma, \calP, S)$ and a string $s \in \Sigma^*$ as input and produces output $\calF_{G,s}$ that acts as an oracle about distance information as follows: for any $A \in \calN$
\begin{itemize}
\item If $A$ minimally and consistently derives $s_{i}^{j}$ with a minimum score $l$, then $\calF_{G,s}(A,i,j)=l$
\item $\calF_{G,s}$ answers queries in constant time.
\end{itemize}
\end{definition}

The above definition is weaker than the local alignment problem, because we are maintaining only those distances for substrings from which the full string can be derived. All known algorithms for parsing and language edit distance computation maintain this information, because not computing these intermediate results may lead to failure in parsing the full string, or parsing it with minimum number of edits.

The choice of an oracle instead of a particular data structure keeps open the possibility that time required for $\laned$ may be $o(n^2)$, which will not be the case if we keep a table like most known parsers. The third condition can be relaxed to take poly-logarithmic time in string and grammar size without much effect.


We reduce {\em distance product computation} over $(min,+)$-structure  to computing language edit distance with insertion. The subcubic equivalence between distance product computation and all-pairs shortest path \cite{williams2010} then establishes Theorem \ref{thm:lowerbound}. If we allow different edit costs for different terminals, then we can allow all three edits: insertion, deletion and substitution.

\paragraph*{Reduction}
We are given two weighted matrix $a$ and $b$ of dimension $m \times m$. We assume weights are all positive integers by scaling and shifting and $M=\max_{i,j}{\Big(a(i,j),b(i,j)\Big)}$. We produce a grammar $G$ and a string $s$ such that from $\calF_{G,s}$ one can deduce the matrix $c=a \cdot b$.

Let us take $d=\lceil m^{1/3} \rceil$, and we set $\delta=d+2$. Our universe of terminals is $\Sigma=\{s_1,s_2,...,s_{3d+6},x\}$. Our input string $s$ is of length $3\delta$ and is simply $s_1s_2....s_{d+2}s_{d+3}...s_{2d+4}s_{2d+5}....s_{3d+6}$.

Now consider a matrix index $i$, $1 \leq i \leq m \leq d^3$. Let
$$f_{1}(i)=\lfloor i/d \rfloor$$ and
$$f_{2}(i)=(i~ mod~ d)+2.$$

Hence $f_1(i) \in [1,d^2]$, and $f_2(i)\in [2,d+1]$. From $f_1(i)$ and $f_2(i)$, we can obtain $i$ uniquely. For notational simplicity we use $i_1$ to denote $f_1(i)$ and $i_2$ to denote $f_2(i)$. Note that if we decompose $s$ into three consecutive equal parts of size $d+2$ each, then $i_2$, $i_2+\delta$ and $i_2+2\delta$ belong to first, second and third halves respectively.

We now proceed to create the grammar $G=\{\calN, \Sigma, \calP, S\}$. Start from $\calN=\{S\}$ and $\calP=\phi$. 
 \begin{itemize}
  \item We create $\lfloor \log{(M+1)} \rfloor$ nonterminals as follows. Let $2^k \leq (M+1) < 2^{k+1}$, then create $X_{2^k}, X_{2^{k-1}},...,X_2,X_1$. We add the productions
$$
 X \longrightarrow x, ~~X_{2^i} \longrightarrow X_{2^{i-1}}X_{2^{i-1}}, 1\leq i \leq k~~ \text{(X-Rule)}.
$$
Let $\hat{w}=X_{2^{j_1}}X_{2^{j_2}}...X_{2^{j_l}}$ if $w=2^{j_1}+2^{j_2}+...+2^{j_l}$.
\item We also add for $1 \leq r \leq d$ the nonterminals $Y_1,Y_2,...,Y_{r}$. Let $2^l \leq r < 2^{l+1}$, then create nonterminals $Z_{2^l},Z_{2^{l-1}},..,Z_1$. Add $$
 Z_1 \longrightarrow \hat{2M+1}, ~~
 Z_{2^{i}} \longrightarrow Z_{2^{i-1}}Z_{2^{i-1}} 1\leq i \leq l ~~\text{(Z-Rule)}.
$$
If $r=2^{j_1}+2^{j_2}+...+2^{j_l}$ add
$$
Y_r \longrightarrow Z_{2^{j_1}}Z_{2^{j_2}}...Z_{2^{j_l}}~~\text{(Y-Rule)}.
$$
\item We now add nonterminal $W$ and productions to generate arbitrary non-empty substrings from $\Sigma \setminus \{x\}$.
$$
W \longrightarrow s_lW|s_l,~~l\in [1,3d+6]~~ \text{(W-Rule)}.
$$
\item We also add nonterminals $W_{i}^{j}$ that generates substring $s_is_{i+1}...s_{j}$.
$$
W_{i}^{j} \longrightarrow s_is_{i+1}...s_{j},~~i,j \in [1,3d+6]~~ \text{(New-W-Rule)}.
$$
\item  Next, we encode the entries of input matrix $A$ and $B$ in our grammar as follows. We add nonterminals from the sets $A_{p,q}, B_{p.q}: 1 \leq p,q,\leq d^2$ and $Y_{r}: 1 \leq r \leq d$. For each entry $a(i,j)=x$, $b(j,k)=y$, we add the production
$$ A_{i_1,j_1}\longrightarrow Y_{\delta-i_2}s_{i_2}W_{i_2+1}^{j_2+\delta-1}\hat{x}s_{j_2+\delta}Y_{j_2}~~  \text{(A-Rule)}, $$
$$B_{j_1,k_1}\longrightarrow Y_{\delta-j_2}s_{j_2+\delta+1}W_{j_2+\delta+2}^{k_2+2\delta-1}\hat{y}s_{k_2+2\delta}Y_{k_2}~~ \text{(B-Rule)}$$

\item We now add nonterminals to combine these consecutive substrings. Add $\{C_{p,q}:1 \leq p,q \leq d^2\}$ and add productions for all $r$, $ 1 \leq r \leq d^2$
$$
C_{p,q}\longrightarrow A_{p,r}B_{r,q}~~ \text{(C-Rule)}
$$
\item Finally, we add the production for the start symbol $S$ for all $p,q$, $1 \leq p,q \leq d^2$
$$
S \longrightarrow WC_{p,q}W~~\text{(S-Rule)}
$$
  \end{itemize}

The following crucial lemma suggests that by looking at $\calF_{G,s}(C_{i_1,j_1}, s_{i_2}^{j_2+2\delta})$ we can derive $c(i,j)$ where $i=(i_1,i_2)$ and $(j_1,j_2)$. This is precisely because $C_{i_1,j_1}$ must derive all the symbols of $s_{i_2}^{j_2}$ exactly once--a property ensured by adding $Y\dash Rule$s, and encodes $c(i,j)$ as the edit distance using $X\dash Rule$s.

\begin{lemma}
\label{lemma:lower1}
For $1 \leq i,j \leq m$, the entry $c_{i,j}=l$, if and only if $C_{i_1j_1}$ minimally and consistently derives $s_{i_2}^{j_2+2\delta}$ with score $l+(2M+1)(2\delta+j_2-i_2)$.
\end{lemma}
\begin{proof}
Fix $i,j$. We first prove the ``only-if'' part. So let $c_{i,j}=l$. Then there must exists a $k$ such that $a_{i,k}=z$ and $b_{k,j}=l-z$.

 We have the C-Rule $C_{i_1,j_1}=A_{i_1,k_1}B_{k_1,j_1}$. Since $a_{i,k}=z$, we have the $A-Rule$ $$A_{i_1,k_1}\longrightarrow Y_{\delta-i_2}s_{i_2}W_{i_2+1}^{k_2+\delta-1}\hat{z}s_{k_2+\delta}Y_{k_2}$$ and since $b_{k,j}=l-z$, we have the $B-Rule$ $$B_{k_1,j_1}\longrightarrow Y_{\delta-k_2}s_{k_2+\delta+1}W_{k_2+\delta+2}^{j_2+2\delta-1}\hat{l-z}s_{j_2+2\delta}Y_{j_2}.$$ Finally, since $i_2+1 < k_2+\delta-1$ and $k_2+\delta+2 \leq j_2+2\delta-1$, $W_{i_2+1}^{k_2+\delta-1} \xRightarrow{\ast}s_{i_2+1}^{k_2+\delta-1}$ and $W_{k_2+\delta+2}^{j_2+2\delta-1} \xRightarrow{\ast}s_{k_2+\delta+2}^{j_2+2\delta}$. All the $x$s generated from $\hat{z}$ and $\hat{(l-z)}$ act as deletion errors which need to be fixed by inserting elements in string $s$. Hence $A_{i_1,k_1}$ derives $s_{i_2}^{k_2+\delta}$ with score $z+(k_2+(\delta-i_2))(2M+1)|$ and $B_{k_1,j_1}$ derives $s_{k_2+\delta+1}^{j_2+2\delta}$ with score $l-z+((\delta-k_2)+j_2)(2M+1)$. Therefore, $C_{i_1,j_1}$ derives $s_{i_2}^{j_2+2\delta}$ with score at most $l+(2\delta+j_2-i_2)(2M+1)$.

Finally, $S \xRightarrow{\ast}s_{1}^{i_2-1}C_{i_1,j_1}s_{j_2+2\delta+1}^{3\delta+6}$ with score at most $l$, since $i_2-1 \geq 1$ and $j_2+2\delta+1< 3\delta+6$, hence $C_{i_1,j_1}$ minimally and consistently derives $s_{i_2}^{j_2+2\delta}$ with score at most $l+(2\delta+j_2-i_2)(2M+1)$.

Now, let us look at the ``if'' part and assume $C_{i_1,j_1}$ derives $s_{i_2}^{j_2+2\delta}$ minimally and consistently with a score $l'$. This can only arise through an application of C-Rule $C_{i_1,j_1}\longrightarrow A_{i_1,k'_1}B_{k'_1,j_1}$ such that $A_{i_1,k'_1}$ derives $s_{i'_2}^{k'_2+\delta}$ within edit distance (say) $z'$ and $B_{k'_1,j_1}$ derives $s_{k''_2+\delta+1}^{j'_2+2\delta}$ within edit distance $l'-z'$. Then, we must have the productions $A_{i_1,k'_1}\longrightarrow Y_{\delta-i'_2}s_{i'_2}W_{i'_2+1}^{k'_2+\delta-1}\hat{z'}s_{k'_2+\delta}Y_{k'_2}$ and $B_{k'_1,j_1}\longrightarrow Y_{\delta-k''_2}s_{k''_2+\delta+1}W_{k''_2+\delta+2}^{j'_2+2\delta-1}\hat{l'-z'}s_{j'_2+2\delta}Y_{j'_2}$. 

First, since we allow only deletion errors, it is not possible that $k'_2 < k''_2$, then edit distance will be $\infty$. Similarly, it is not possible that $i'_2 > i_2$ or $j'_2 < j_2$.

Therefore, the total edit cost paid is $\left[(\delta-i'_2)+k'_2+(\delta-k''_2)+j_2\right](2M+1)+\alpha$ for some $\alpha \leq 2M$. If $i'_2 > i_2$ or $j'_2 < j_2$, then the above cost is always higher than the case when $i'_2=i_2$ and $j'_2=j_2$. Hence we must use the productions $A_{i_1,k'_1}\longrightarrow Y_{\delta-i_2}s_{i_2}W_{i_2+1}^{k'_2+\delta-1}\hat{z'}s_{k'_2+\delta}Y_{k'_2}$ and $B_{k'_1,j_1}\longrightarrow Y_{\delta-k''_2}s_{k''_2+\delta+1}W_{k''_2+\delta+2}^{j_2+2\delta-1}\hat{l'-z'}s_{j_2+2\delta}Y_{j_2}$.

 If $k'_2 > k''_{2}$, then the total score is $\geq (\delta+1)(2M+1)$ due to the $Y$ rules, which is always higher than the case when $k'_2=k''_2$. Therefore, it must happen that $k'_2=k''_2$. But this can only happen, if there is a number $k'$ such that $f_1(k')=k'_1$ and $f_2(k')=k'_2$ and $a(i,k')=z'-[(\delta-i'_2)+k'_2](2M+1)$ and $b(k',j)=l'-z'-[(\delta-k'_2)+j'_2](M+1)$, and therefore $c(i,j)\leq l'-(2\delta+j'_2-i'_2)(2M+1)$.

The lemma now follows.
\end{proof}
\paragraph*{Grammar Size}
The total number of nonterminals used in this grammar is $O(d^4+d^2+\log{M})=O(m^{4/3}+\log{M})$ and the number of productions is $O(m^2+d^6+d^2+\log{M})=O(m^2+\log{M})$, where $d^6$ term comes from the $C\text{-}Rule$ and $m^2$ comes from considering all the entries of $A$ and $B$. If we consider the number of nonterminals involved in each production, then the total size of the grammar is $|G|=O(m^2\log{M})$.

{\bf Note.} The grammar constructed here is not in CNF form, but can easily be transformed into a CNF representation $G'$ where the number of productions in $G'$ increases at most by a factor of $\log{m}$. This happens because in $G$ there is no $\epsilon$ production or unit productions. For every terminal $s_j$ $j=1,2,..,3d+6$, we create a nonterminal, $S_j$ and replace their occurrences in productions with the newly created nonterminals. We add the productions $S_j \rightarrow s_j$ for $j=1,2,..,3d+6$. Finally, for every rule of the form $Q\rightarrow R_1R_2...R_s$, we create $s-1$ rules $Q \rightarrow R_1Q_1$, $Q_1 \rightarrow R_2Q_2$,..., $Q_{s-2} \rightarrow R_{s-1}R_{s}$. Since in $G$, the size of $RHS$ is any production can be at most $\lceil \log{M} \rceil+3$, we get the desired bound. Therefore, the claims in this section equally holds when parsers are restricted to work with CNF grammars.

\paragraph*{Time Bound}
\begin{lemma}
\label{lemma:dp}
Any language edit distance problem $P$ with mc-derivation having run time $O(T(|G|)t(n))$ on grammars of size $|G|$ and strings of
length $n$ can be converted into an algorithm MP to compute distance product on positive-integer weighted $m \times m$ matrix with highest weight $M$ that runs in time $O(max(m^2+T(m^2)t(m^{1/3})\log{M})$. In particular, if P takes $O(|G|n^{3-\epsilon})$ time then that implies an $O(m^{3-\epsilon/3}\log{M})$ running time for MP.
\end{lemma}
\begin{proof}
Given the two matrices $A$ and $B$ of dimension $m \times m$ with maximum weight (after shifting and scaling) $M$, time to read the entries is $O(m^2)$ and to create grammar $G$ is $O(|G|)=O(m^2+m^{4/3}\log{M})$ (note that $d=\lfloor m^{1/3} \rfloor$) and string $s$ is $O(m^{1/3})$. Assume, the parser takes time $O(T_1(G)+T_2(G)t(n))$ to create $\calF_{G,s}$. Then we query $\cal{F}_{G,s}$ for each $c_{i,j}$ by creating the query $(C_{i_1,j_1},s_{i_2}^{j_2+2\delta})$. If the answer is $K$, we set $c_{i,j}=K$. By Lemma \ref{lemma:lower1}, the computed value of $c_{i,j}=\min_{k}(a_{i,k}+b_{k,j})$ is correct. Hence once parsing has been done, creating $C$ again takes $O(m^2)$ time, assuming each query needs $O(1)$ time.

Suppose $T_1(G)=T_2(G)=|G|$ and $t(n)=n^{3-\epsilon}$, $ 0 < \epsilon \leq 1$then we get an algorithm to compute distance product in time $O(m^{3-\frac{\epsilon}{3}}\log{M})$.
\end{proof}

Now, due to sub-cubic equivalence of distance product computation with APSP, Theorem \ref{thm:lowerbound} follows.

\subsection*{Reducing APSP to Stochastic Context Free Parsing}
\sloppy
The reduction takes the following steps.
\begin{enumerate}
\item Reduce $(\min, \times)$-matrix product where matrix entries are drawn from $\mathbb{R}^{+}$ to stochastic context free grammar parsing, that is show if there exists am $O(|G|n^{3-\epsilon}\max_{p \in {\bf p}}{\log{\frac{1}{p}}})$ algorithm for stochastic context free grammar parsing, then there exists one with running time $\tilde{O}(n^{3-\beta}\log{W})$ for $(\min, \times)$-matrix product over $\mathbb{R}^{+}$, $\epsilon, \beta >0$ where $W$ is the maximum weight of any entry.
\item Next we show if there exists an algorithm with running time $\tilde{O}(n^{3-\beta}\log{W})$ for $(\min, \times )$-matrix product with entries in $(0,W]$, $\beta >0$, then there exists one with running time $\tilde{O}(n^{3-\beta}\log{W})$ for detecting negative weight triangle in a weighted graph with weights ranging in $[-W,W]$.
\item Finally, due to sub-cubic equivalence between minimum weight triangle detection with non-negative weights and APSP \cite{williams2010}, the result follows.
\end{enumerate}

\paragraph*{Reducing $(\min, \times)$-matrix product with entries in $\mathbb{R}^{+}$ to stochastic context free grammar parsing}

This reduction is similar to the previous one used for reducing language edit distance problem to distance product computation. Instead of encoding $a_{i,j}=w$, in the production rules $A\dash Rule$ and $B\dash Rule$, this is encoded in the probability of the corresponding productions. 

\begin{definition}
Given a stochastic context free grammar $\{G=(\calN,\Sigma, \calP, S), {\bf p}\}$, and a string $s \in \Sigma^*$. A nonterminal $A \in \calN$ c-derives (consistently derives) $s_i^j$ if and only if the following condition holds:

1. $A$ derives $s_{i}^{j}$

2. There is a derivation sequence $S \xRightarrow{\ast} s_{1}^{i-1}As_{j+1}^{n}$.
\end{definition}

\begin{definition}
A $Stochastic\dash Parsing$~is an algorithm that takes a SCFG $\{G=(\calN,\Sigma, \calP, S),{\bf p}\}$ and a string $s \in \Sigma^*$ as input and produces output $\calF_{\{G,{\bf p}\},s}$ that acts as an oracle about distance information as follows: for any $A \in \calN$
\begin{itemize}
\item If $A$ consistently derives $s_{i}^{j}$ with maximum probability $q$, then $\calF_{\{G,{\bf p}\},s}(A,i,j)=q$
\item $\calF_{\{G,{\bf p}\},s}$ answers queries in constant time.
\end{itemize}
\end{definition}

\paragraph*{Creating the Grammar.}

We are given two matrices $a$ and $b$ with entries from $\mathbb{R}^{+}$, and want to compute their $(\min, \times)$-product $c=a.b: c_{i,j}=\min_{1 \leq k \leq n}(a_{i,k}.b_{k,j}).$

\noindent {\bf Input string.} Let us take $d=\lceil m^{1/3} \rceil$, and we set $\delta=d+2$. Our universe of terminals is $\Sigma=\{s_1,s_2,...,s_{3d+6},x\}$. Our input string $s$ is of length $3\delta$ and is simply $s_1s_2....s_{d+2}s_{d+3}...s_{2d+4}s_{2d+5}....s_{3d+6}$.
\noindent {\bf Grammar construction.} Consider a matrix index $i$, $1 \leq i \leq m \leq d^3$. Let $f_{1}(i)=\lfloor i/d \rfloor$ and
$f_{2}(i)=(i~ mod~ d)+2$. Hence $f_1(i) \in [1,d^2]$, and $f_2(i)\in [2,d+1]$. We can uniquely obtain $i$ from $f_1(i)$ and $f_2(i)$. For notational simplicity we use $i_1$ to denote $f_1(i)$ and $i_2$ to denote $f_2(i)$. If we decompose $s$ into three consecutive equal parts of size $d+2$ each, then $i_2$, $i_2+\delta$,  and $i_2+2\delta$ belong to first, second and third halves respectively.
We now create the grammar $(G, {\bf p})=(\{\calN, \Sigma, \calP, S\}, {\bf p})$. Start from $\calN=\{S\}$ and $\calP=\phi$. 
 \begin{itemize}
  \item We add nonterminal $W$ and productions to generate arbitrary non-empty substrings from $\Sigma \setminus \{x\}$.
$
W \longrightarrow s_lW|s_l,~l\in [1,3d+6] \text{ each rule having probability } \frac{1}{2(3d+6)} ~~\text{(W-Rule)}.
$
The probabilities add up to $1$.
\item We encode the entries of $A$ and $B$ in our grammar. We add nonterminals $A_{p,q}: 1 \leq p,q,\leq d^2$, and $B_{p,q}: 1 \leq p,q \leq d^2$. Let $Count_A(i_1,j_1)=\sum_{i,j: x=f_1(i), y=f_1(j)}\frac{1}{a_{i,j}}$, and $Count_B(i_1,j_1)=\sum_{i,j: x=f_1(i), y=f_1(j)}\frac{1}{b_{i,j}}$.  Set
$MaxCount_A=\max_{x,y, 1 \leq x,y \leq d^2} Count_A(i_1,j_1)$, and 
$MaxCount_B=\max_{x,y, 1 \leq x,y \leq d^2}Count_B(i_1,j_1)$.
For each entry $a_{i,j}, b_{i,j}$, $i=(i_1,i_2)$, and $j=(j_1,j_2)$ 
add productions (A-Rule):
$
A_{i_1,j_1}\longrightarrow s_{i_2}Ws_{j_2+\delta} \text{ with prob. $\frac{1}{a_{i,j}MaxCount_A}$}, 
$
if $MaxCount_A > Count_A(i_1,j_1)$, then add a dummy rule
$
A_{i_1,j_1}\longrightarrow x \text{ with prob. $\frac{MaxCount_A-Count_A(i_1,j_1)}{MaxCount_A}$}
.$ Add productions (B-Rule) by replacing every ``A'' and ``a'' in (A-Rule) with ``B'' and ``b'' respectively.
\item We add nonterminals $\{C_{p,q}:1 \leq p,q \leq d^2\}$ and the productions for all $r$, $ 1 \leq r \leq d^2$
$
C_{p,q}\longrightarrow A_{p,r}B_{r,q} \text{ with prob. $\frac{1}{d^2}$}~~ \text{(C-Rule)}
$
\item Finally, we add the production for the start symbol $S$ for all $p,q$, $1 \leq p,q \leq d^2$
$
S \longrightarrow WC_{p,q}W \text{ with prob. $\frac{1}{d^4}$}~~ \text{(S-Rule)}
$
\end{itemize}
It can be verified that probabilities of all rules with same nonterminal on the LHS add up to $1$. Hence the constructed grammar is a SCFG. 
The following lemma suggests that by looking at $\calF_{G,s}(C_{i_1,j_1}, s_{i_2}^{j_2+2\delta})$ we can derive $c(i,j)$ where $i=(i_1,i_2)$ and $(j_1,j_2)$. Then noting that the grammar size is $O(m^2)$, and string length $O(m^{1/3})$, we get the desired subcubic equivalence between SCFG parsing and $(\min, \times)$ matrix product (Lemma \ref{lemma:dp1}). Note that this non-CNF grammar can easily be converted into a CNF representation with constant factor blow-up in size.

\begin{lemma}
\label{lemma:lower2}
For $1 \leq i,j \leq m$, the entry $c_{i,j}=l$, if and only if $C_{i_1j_1}$ $c$-derives $s_{i_2}^{j_2+2\delta}$ with probability $\frac{1}{ld^2.MaxCount_A.MaxCount_B[2(3d+6)]^{j_2+2\delta-i_2-2}}$\end{lemma}
\begin{proof}
The proof is similar to Lemma \ref{lemma:lower1} instead of computing the total edit distance, compute the total probability of the productions applied to parse $s_{i_2}^{j_2+2\delta}$. 

Fix $i,j$. We first prove the ``only-if'' part. So let $c_{i,j}=l$. Then there must exists a $k$ such that $a_{i,k}=z$ and $b_{k,j}=\frac{l}{z}$.

 We have the C-Rule $C_{i_1,j_1}=A_{i_1,k_1}B_{k_1,j_1}$ with probability $\frac{1}{d^2}$. Since $a_{i,k}=z$, we have the (A-Rule) $A_{i_1,k_1}\longrightarrow s_{i_2}Ws_{k_2+\delta}$ with probability $\frac{1}{zMaxCount_A}$and since $b_{k,j}=\frac{l}{z}$, we have the (B-Rule) $B_{k_1,j_1}\longrightarrow s_{k_2+\delta+1}Ws_{j_2+2\delta}$ with probability $\frac{z}{lMaxCount_B}$ . Finally, since $i_2+1 < k_2+\delta-1$ and $k_2+\delta+2 \leq j_2+2\delta-1$, $W \xRightarrow{\ast}s_{i_2+1}^{k_2+\delta-1}$ with probability $\frac{1}{[2(3d+6)]^{k_2+\delta-i_2-1}}$ and $W \xRightarrow{\ast}s_{k_2+\delta+2}^{j_2+2\delta}$ with probability $\frac{1}{[2(3d+6)]^{j_2+2\delta-k_2-\delta-1}}$. Hence  $C_{i_1,j_1}$ derives $s_{i_2}^{j_2+2\delta}$ with probability at least $\frac{1}{d^2}\frac{1}{zMaxCount_A}\frac{z}{lMaxCount_B}\frac{1}{[2(3d+6)]^{k_2+\delta-i_2-1+j_2+2\delta-k_2-\delta-1}}=\frac{1}{d^2}\frac{1}{l MaxCount_A MaxCount_B}\frac{1}{[2(3d+6)]^{j_2+2\delta-i_2-2}}$.

Finally, $S \xRightarrow{\ast}s_{1}^{i_2-1}C_{i_1,j_1}s_{j_2+2\delta+1}^{3\delta+6}$ with probability $\frac{1}{d^4}$, since $i_2-1 \geq 1$ and $j_2+2\delta+1< 3\delta+6$, hence $C_{i_1,j_1}$ consistently derives $s_{i_2}^{j_2+2\delta}$.

Now, let us look at the ``if'' part and assume $C_{i_1,j_1}$ derives $s_{i_2}^{j_2+2\delta}$  consistently with a probability $q$. This can only arise through an application of C-Rule $C_{i_1,j_1}\longrightarrow A_{i_1,k'_1}B_{k'_1,j_1}$ with probability $\frac{1}{d^2}$ such that $A_{i_1,k'_1}$ derives $s_{i_2}^{k_2+\delta}$ and $B_{k'_1,j_1}$ derives $s_{k_2+\delta+1}^{j_2+2\delta}$. Then, we must have the productions $A_{i_1,k'_1}\longrightarrow s_{i_2}Ws_{k_2+\delta}$ with probability $\frac{1}{zMaxCount_A}$ where $z=a(i,k), i=(i_1,i_2), k=(k'_1,k_2)$ and $B_{k'_1,j_1}\longrightarrow s_{k'_2+\delta+1}Ws_{j_2+2\delta}$ with probability $\frac{1}{z'MaxCount_B}$ where $z'=b(k,j), k=(k'_1,k_2), j=(j_1,j_2)$. Now considering the probabilities of W-Rules to generate $s_{i_2+1}^{k_2+\delta}$ and $s_{k'_2+\delta+2}^{j_2+2\delta-1}$, the ``if'' part is established.

The lemma now follows.
\end{proof}

\begin{lemma}
\label{lemma:dp1}
Any stochastic context free parsing problem $P$ with c-derivation having run time $O(T(|G|)t(n)\max{\log_{p \in \bf{p}}})$ on grammars of size $|G|$ and strings of
length $n$ can be converted into an algorithm MP to compute $(\min,\times)$-product of $m \times m$ matrices with entries in $\mathbb{R} \setminus \{0\}$ with highest weight $M$ that runs in time $\tilde{O}(max(m^2+T(m^2)t(m^{1/3})\log{M})$. In particular, if P takes $O(|G|n^{3-\epsilon}\max{\log_{p \in \bf{p}} \frac{1}{p}})$ time then that implies an $\tilde{O}(m^{3-\epsilon/3}\log{M})$ running time for MP.
\end{lemma}
\begin{proof}
The proof follows from Lemma \ref{lemma:lower2} and following similar steps as Lemma \ref{lemma:dp}.
\end{proof}

\paragraph*{From $(min,\times)$-matrix product with positive real entries to Negative Triangle Detection.}

we show if there exists an algorithm with running time $\tilde{O}(n^{3-\beta}\log{W})$ for $(\min, \times )$-matrix product with entries in $(0,W]$, $\beta >0$, then there exists one with running time $\tilde{O}(n^{3-\beta}\log{W})$ for detecting negative weight triangle in a weighted graph with weights ranging in $[-W,W]$.

We now show that if there exists an algorithm with running time $\tilde{O}(n^{3-\beta}\log{W})$ for $(\min, \times )$-matrix product over $\mathbb{R}^{+}$ with weights in $(0,W)$, $\beta >0$, then there exists one with running time $\tilde{O}(n^{3-\beta}\log{W})$ to detect if a weighted graph $G=(V,E)$ has a triangle of negative total edge weight where weights are in $[-W,W]$.

We assume all edge weights are integers, and the maximum absolute weight $W$ is at least 3. Both of these can be achieved by appropriately scaling the edge weights.

Let $W$ be the maximum absolute weight on any edge $e \in E$, $|W| \geq 3$. Set $A(i,j)=w_{i,j}+W^3$, and $B(i,j)=A(i,j)$. Therefore, all entries of $A$ and $B$ are $> 0$. Find the $(\min, \times)$ product of $C=A\bigodot_{\min,\times}B$. Let $C'(i,j)=C(i,j)-W^6+W^3w_{i,j}+2W^2$.

If there exists a negative triangle $i\rightarrow k\rightarrow j$, then $w_{i,k}+w_{k,j}+w_{i,j} \leq -1$. Hence $W^3(w_{i,k}+w_{k,j}+w_{i,j}) \leq -W^3$ or, $ W^3(w_{i,k}+w_{k,j}+w_{i,j})+2W^2\leq -W^3+2W^2\leq -W^2$. Now $(w_{i,k}+W^3)(w_{k,j}+W^3)=w_{i,k}w_{k,j}+W^3(w_{i,k}+w_{k,j})+W^6$ Hence, $C'(i,j)=\min_{k}( w_{i,k}w_{k,j}+W^3(w_{i,k}+w_{k,j}+w_{i,j})+2W^2$. Now $-W^2\leq w_{i,k}w_{k,j} \leq W^2$. Therefore, if there is a negative triangle involving edge $(i,j)$, then
$C'(i,j) \leq W^2+\min_{k}(W^3(w_{i,k}+w_{k,j}+w_{i,j})+2W^2\leq 0$

On the other hand, if there is no negative triangles, then $C'(i,j)\geq -W^2+2W^2=W^2 \geq 9$ for all $1 \leq i,j \leq n$.

Therefore, there exists a negative triangle in $G$ if and only if there is a negative entry in $C'$. While $C$ can be computed in asymptotically same time as computing $(\min, \times)$-matrix product of two $n \times n$ dimensional matrices with real positive entries, $C'$ can be computed from $C$ in $O(n^2)$ time.

Hence, we get the following lemma.

\begin{lemma}
Given two $n \times n$ matrices with positive real entries with maximum weight $W$ if their $(\min, \times)$-matrix product can be done in time $O(T(n)\log{W})$ time, then $\tilde{O}([T(n)+n^2]\log{W})$ time is sufficient to detect negative triangles on weighted graphs with $n$ vertices and weights in $[-W,W]$.
\end{lemma}

Now, due to subcubic equivalence between negative triangle detection and APSP, we get the following theorem.

\begin{theorem}\label{theorem:lowerb-scfg}
Given a stochastic context-free grammar $\{G=(\calN,\Sigma, \calP, S),{\bf p}\}$, and a string $s \in \Sigma^*$, $|s|=n$, if the SCFG parsing problem can be solved in $O(|G|n^{3-\delta}\max{\log_{p \in \bf{p}} \frac{1}{p}})$ time then that implies an algorithm with running time $\tilde{O}(m^{3-\delta/3}\log{W})$ for all-pairs shortest path problem on weighted digraphs with $m$ vertices and maximum weight $W$.
\end{theorem}

This leads to the following corollary by sub-cubic equivalence of all-pairs shortest path with many other fundamental problems on graphs and matrices \cite{williams2010,abboud:14}.
\begin{corollary}
\label{cor:lb}
Given a stochastic context-free grammar $\{G=(\calN,\Sigma, \calP, S),{\bf p}\}$, and a string $s \in \Sigma^*$, $|s|=n$, if the SCFG parsing problem can be solved in $O(|G|n^{3-\delta}\max{\log_{p \in \bf{p}} \frac{1}{p}})$ time then that implies an algorithm with running time $\tilde{O}(m^{3-\delta/3})$, $\gamma, \delta >0$ for all of the following problems.
\begin{enumerate}
        \item Minimum weight triangle: Given an $n$-node graph with real edge weights, compute $u,v,w$ such that $(u,v),(v,w),(w,u)$ are edges and the sum of edge weights is minimized.
        \item Negative weight triangle: Given an $n$-node graph with real edge weights, compute $u,v,w$ such that $(u,v),(v,w),(w,u)$ are edges and the sum of edge weights is negative.
        \item Metricity: Determine whether an $n \times n$ matrix over $\mathbb{R}$ defines a metric space on $n$ points.
        \item Minimum cycle: Given an $n$-node graph with real positive edge weights, find a cycle of minimum total edge weight.
        \item Second shortest paths: Given an $n$-node directed graph with real positive edge weights and two nodes $s$ and $t$, determine the second shortest simple path from $s$ to $t$.
        \item Replacement paths: Given an $n$-node directed graph with real positive edge weights and a shortest path $P$ from node $s$ to node $t$, determine for each edge $e \in P$ the shortest path from $s$ to $t$ in the graph with $e$ removed.
        \item Radius problem: Given an $n$-node weighted graph with real positive edge weights, determine the minimum distance $r$ such that there is a vertex $v$ with all other vertices within distance $r$ from $v$.
\end{enumerate}
\end{corollary}

\subsection*{Reducing APSP to Weighted Language Edit Distance Problem}
In the weighted language edit distance problem, we are given a context free language $\calL(G)=(\calN, \Sigma, \calP, S)$ and a string $s \in \Sigma^*$ along with a scoring function $score: \Sigma \times \{Insertion, Deletion, Substitution\} \rightarrow \mathbb{R}^{+}$, the goal is to do minimum total weighted edits on $s$ according to the scoring function to map it to $\calL(G)$.

To reduce APSP to weighted language edit distance problem, we use the same construction used to prove Theorem \ref{thm:lowerbound}, and in addition we define a scoring function. For insertion edits, all terminals in $\Sigma$ get a score of $1$. However, for deletion and substitution, we set for every $x \in \Sigma$, $score(x,deletion)=score(x,substitution)=(3d+6)(M+1)$. Deletion or substitution of any element in the input string $s$ is too costly. An optimum algorithm for the weighted language edit distance will therefore never do deletion or substitution. The entire analysis of Theorem \ref{thm:lowerbound} now applies.


\end{document}